\documentclass[letterpaper,twocolumn,10pt]{article}
\usepackage{usenix2019_v3}

\usepackage{tikz}
\usepackage{amsmath}
\usepackage{amsmath}
\usepackage{hyperref}
\usepackage[capitalize,noabbrev]{cleveref}
\usepackage{url}
\usepackage{bm}
\usepackage{ulem}
\usepackage{amssymb}
\usepackage{tabularx}
\usepackage{multicol}
\usepackage{booktabs}
\usepackage{multirow}
\usepackage{tikz}
\usepackage{ulem}
\usepackage{amsthm}
\usepackage{enumerate}
\usepackage{enumitem}
\usepackage{subfigure}
\usepackage[graphicx]{realboxes}
\usepackage{stfloats}
\usepackage{xcolor}
\usepackage{ulem}
\usepackage{caption}
\usepackage{makecell}
\usepackage{setspace}
\usepackage{bbding}
\usepackage[ruled]{algorithm2e}
\usepackage{ulem} 
\usetikzlibrary{spy}
\usetikzlibrary {decorations.fractals,spy}
\usepackage{filecontents}
\usepackage{colortbl}
\usepackage[most]{tcolorbox}
\usepackage{graphicx} 
\usepackage{wrapfig}

\newcommand{\sysname}{PVF}
\newcounter{mycounter}[section]
\newtheorem{definition}[mycounter]{Definition}
\newtheorem{theorem}{Theorem}
\newtheorem{remark}{Remark}
\newtheorem{lemma}[mycounter]{Lemma}
\newtheorem{example}{Example}
\newtcolorbox{empheqboxed}{colback=blue!6, 
 colframe=white,
 width=\linewidth,
 sharpish corners,
 top=1mm, %
 bottom=0pt,
 left=2pt,
 right=2pt
}
\usepackage{filecontents}
\usepackage{xcolor}
\definecolor{deepgreen}{rgb}{0, 0.5, 0}

\begin{document}


\title{\Large \bf $\lambda$-SecAgg: Partial Vector Freezing for Lightweight Secure Aggregation \\ in Federated Learning}

\author{
{\rm Siqing Zhang}\\
University of Science and Technology of China\\
Hefei, China\\
Email: siqingzhang@mail.ustc.edu.cn
\and
{\rm Yong Liao}\\
University of Science and Technology of China\\
Hefei, China\\
Email: yliao@ustc.edu.cn
\and
{\rm Peng Yuan Zhou}\\
Aarhus University\\
Aarhus, Denmark\\
Email: pengyuan.zhou@ece.au.dk
} 

\maketitle

\begin{abstract}
Secure aggregation of user update vectors (e.g. gradients) has become a critical issue in the field of federated learning. Many Secure Aggregation Protocols (SAPs) face exorbitant computation costs, severely constraining their applicability. Given the observation that a considerable portion of SAP's computation burden stems from processing each entry in the private vectors, we propose \textbf{P}artial \textbf{V}ector \textbf{F}reezing (\textbf{PVF}), a portable module for compressing computation costs without introducing additional communication overhead. \textbf{$\bm{\lambda}$-SecAgg}, which integrates SAP with PVF, ``freezes'' a substantial portion of the private vector through specific transformations, requiring only $\frac{1}{\lambda}$ of the original vector to participate in SAP. Eventually, users can ``thaw'' the public sum of the ``frozen entries'' by the result of SAP. To avoid potential privacy leakage, we devise Disrupting Variables Extension for PVF.
We demonstrate that PVF can seamlessly integrate with various SAPs and it poses no threat to user privacy in the semi-honest and active adversary settings. We include $7$ baselines, encompassing $5$ distinct types of masking schemes, and explore the acceleration effects of PVF on these SAPs. Empirical investigations indicate that when $\lambda=100$, PVF yields up to $99.5\times$ speedup and up to $32.3\times$ communication reduction. 
\end{abstract}

\section{Introduction}
Machine learning technologies are applied in countless fields to improve service performance. However, aggregating large amounts of data for big data mining raises concerns regarding data privacy~\cite{liu2021machine}. \textit{Federated Learning} (FL)~\cite{mcmahan2017communication} keeps original data on the local devices while only requiring data owners to submit local training updates to a central server. Nonetheless, as~\cite{zhu2019deep} and \cite{geiping2020inverting} indicate, attackers can infer a user's local data by reversing the submitted updates. To address this issue, numerous research efforts have been focusing on \textit{Secure Aggregation Protocols} (SAPs)~\cite{liu2022privacy} for aggregating all user's model information while preserving individual privacy.

The widely discussed SAPs are based on \textit{Secure Multi-Party Computation} (SMPC)~\cite{xu2022non,Sotthiwat2021PartiallyEM}, \textit{Mask}~\cite{Bonawitz2017PracticalSA}, \textit{Homomorphic Encryption} (HE)~\cite{aono2017privacy} and \textit{Differential Privacy} (DP)~\cite{wei2020federated}. For most SAPs, except for DP-based ones, the computation overhead often scales proportionally with the length of the model update vectors since most of these schemes involve masking (encrypting) each entry of the vector sequentially. 
Therefore the \textbf{computation} time for both masking and unmasking always experiences a steep escalation with the increase in vector length, as PracAgg~\cite{Bonawitz2017PracticalSA} in Figure~\ref{ex0}, significantly constraining real-world applications. Especially in recent applications that utilize FL to fine tune \textit{Large Language Models} (LLMs)~\cite{ye2024openfedllm} with billions of parameters, the computational overhead brought by SAP is unbearable.

\begin{figure}[t!]
\centering
\includegraphics[width=\linewidth]{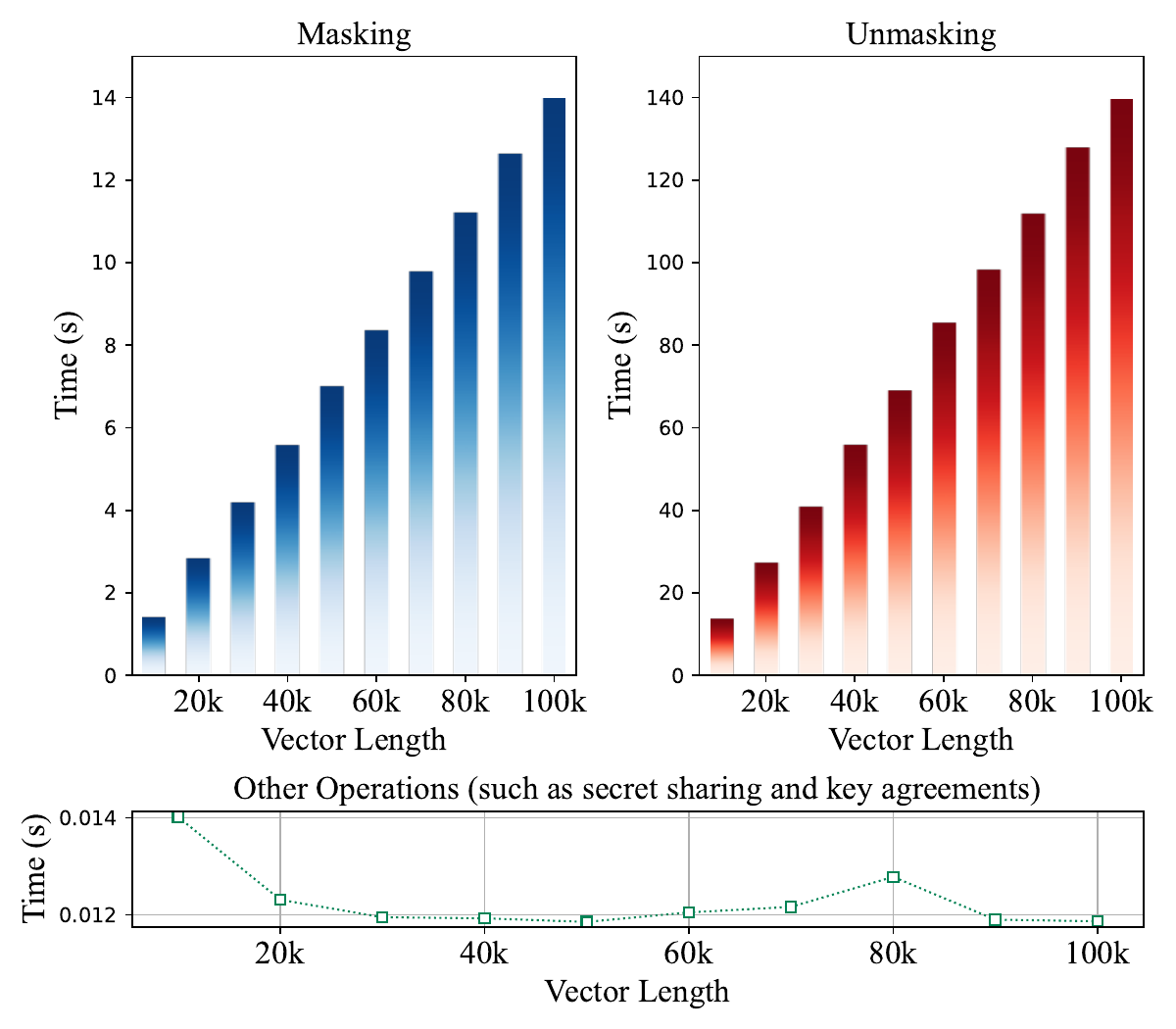}
\caption{Computation time of various operations in PracAgg for a single aggregation across different vector lengths, with 100 users and 10\% dropout rate.}
\label{ex0}
\end{figure}
On the other hand, DP-based solutions although have the best efficiencies, many studies~\cite{stevens2022efficient,wang2021protecting} posit that the minimal noise added by DP is insufficient to thwart attacks such as gradient inversion aimed at stealing users' local data, where adversaries can recover flawed but recognizable handwritten digit image~\cite{wang2021protecting}. Therefore the security of DP in secure aggregation faces challenges, necessitating its combination with masking to bolster privacy. In this work, we primarily explore methods to alleviate \textbf{masking-related} overhead.

\textbf{Our proposal.}
We think the root cause of computation overhead in SAPs is masking each entry of the original vector. While a few \textit{sparsification}-related approaches~\cite{Ergun2021SparsifiedSA,lu2023top} try to reduce the dimensions of uploaded vectors, they raise an inevitable trade-off of discarding some information. 
In this work, we propose \textbf{P}artial \textbf{V}ector \textbf{F}reezing (\textbf{PVF}) to reduce the number of entries processed in SAPs while ensuring intact aggregation of all entries in the original vector. Within the module, each user performs certain transformations on the original vector at a negligible computation cost to selectively \textit{freeze} most entries of the user’s original vector, compressing the length of the vector involved in SAPs to $\frac{1}{\lambda}$ of its original size ($\lambda$ is the compression factor). The communication overhead and number of communication interactions in SAPs after integrating PVF, which we call \textbf{$\bm{\lambda}$-SecAgg}, do not increase. Further, we propose \textit{Disrupting Variables Extension} to prevent PVF from leaking the linear relationship between vectors.
\sysname\ remains decoupled from SAPs and guarantees individual user privacy under semi-honest and active adversary settings, offering high portability. 
  

Our contributions can be summarized as follows:
\begin{itemize}[leftmargin=*]
    \item We propose the PVF without incurring additional communication overhead or harming security. It reduces the entries processed in SAP while \textbf{ensuring the aggregation of all entries in the original vector}, which means it can compress the computation overhead of SAP to approximately $\frac{1}{\lambda}$ of the original. Moreover, it brings up to $32.3\times~(\lambda=100)$ additional communication enhancements for HE-based SAPs thanks to the decreased number of ciphertext entries.
    \item We propose the disrupting variables extension to \sysname\ to avoid potential privacy leakage.
    \item Extensive experiments show the effectiveness of our proposal. We include $7$ baselines encompassing $5$ types of \textbf{masking schemes} for a comprehensive overhead comparison, which is largely unexplored in most research endeavors and reaffirms the high portability of \sysname.
\end{itemize}


\section{Related Work}
\label{sec:related}
\textbf{Secure Aggregation Protocols.} Various types of SAPs have been proposed, including SMPC-based~\cite{boer2020secure,xu2022non}, HE-based~\cite{aono2017privacy,ma2022privacy,Li2022EfficientPF}, DP-based~\cite{Geyer2017DifferentiallyPF,wei2020federated}, and Mask-based~\cite{Bonawitz2017PracticalSA} schemes. Most efforts to reduce computation cost focus on enhancing PracAgg~\cite{Bonawitz2017PracticalSA}, which is mainly categorized into two types: (i) improving the masking mechanism~\cite{Liu2022EfficientDA,stevens2022efficient,liu2022sash,Wei2023LightweightFL} to reduce \textbf{masking-related} overhead; (ii) minimizing \textbf{interaction-related} overhead, including refining communication structures~\cite{Bell2020SecureSA,so2021turbo} and enhancing efficiency in key agreements among users~\cite{kalikinkar2018nikebased,kadhe2020fastsecagg,Ma2023FlamingoMS}.
\textbf{Note} that the security of FL remains an open issue. SAPs, though cannot fully guarantee FL security at the moment~\cite{elkordy2023much}, remain a promising direction worth exploration. The main objective of our work is to reduce the masking-related overhead of secure aggregation, thereby making it more applicable in practice.

\textbf{Compression-based Techniques}. \cite{rothchild2020fetchsgd} employs a \textit{Count Sketch} to compress model updates. Additionally, some sparsification-based approaches~\cite{Ergun2021SparsifiedSA,lu2023top} can reduce vector dimensions. Our method fundamentally differs from these schemes, as our proposal compresses the entries involved in secure aggregation while retaining the \textbf{intact} aggregation result of all original entries.

\textbf{Defense Against Malicious Server}. Several works indicate the malicious server can launch \textit{Model Inconsistency Attacks}~\cite{pasquini2022eluding}, \textit{Multi-round Privacy Stealing Attacks}~\cite{so2023securing} and \textit{Aggregation Falsification Attacks}~\cite{Guo2021VeriFLCA}. These studies also propose strategies to counter these attacks accordingly, only requiring minor modifications to the SAP process, as described in Section~\ref{portability}.

\textbf{Input constraints}. Several works~\cite{bell2023acorn,lycklama2023rofl} are proposed to mitigate \textit{Poisoning Attacks} in FL. They delineate that the erroneous inputs of malicious users can result in the server obtaining an inaccurate global model, thereby harming the training task. They propose methodologies utilizing \textit{Zero-Knowledge Proofs} to bound user inputs. However, their ability to prevent poisoning attacks is limited~\cite{Ma2023FlamingoMS}. Establishing strong constraints against malicious inputs remains an unresolved challenge, and it falls beyond the scope of this work. What's more, \cite{Mozaffari2023EveryVC} propose \textit{Federated Rank Learning} (FRL), where the server aggregates the parameter rankings instead of the model parameter updates. It can effectively resist poisoning attacks, and enable direct aggregations without any constraints on user submissions. Therefore, FRL can be combined with SAP, and we do not have to worry about whether \sysname\ can be integrated with input constraints in this work.

\section{Cryptographic Primitives}
\subsection{Symmetric Authenticated Encryption}
\label{section:ae}
Symmetric authenticated encryption can ensure the confidentiality of a message, including:
\begin{itemize}[leftmargin=*]  
    \item $AE.gen(k)\rightarrow(sk)$, where $k$ is the security parameter. It outputs a secret key $sk$.
    \item $AE.enc(sk,m)\rightarrow(c)$. It encrypts the message $m$ using $sk$ and outputs the ciphertext $c$.
    \item $AE.dec( {sk,c} )\rightarrow m~or~\bot$. If $sk$ is the correct key corresponding to the ciphertext $c$ and $c$ passes integrity verification, it outputs the plaintext $m$. Otherwise, it outputs an error symbol.
\end{itemize}
We need the encryption scheme to be indistinguishable under chosen plaintext attacks (IND-CPA) and ciphertext integrity (IND-CTXT)~\cite{bellare2000authenticated}. 

\subsection{Digital Signature}
Digital signature can ensure the authenticity and integrity of a message. We use the signature scheme that achieves security against universal forgery under chosen message attack (UF-CMA). The digital signature scheme consists of:
\begin{itemize}[leftmargin=*]  
    \item $DS.gen(k)\rightarrow(sk,pk)$, where $k$ is the security parameter. It outputs a secret key $sk$ and a public key $pk$.
    \item $DS.sign(sk,m)\rightarrow(sig)$. It outputs a digital signature $sig$ on the message $m$. 
    \item $DS.verify( {sig,pk,m} )\rightarrow True~or~False$. It verifies whether the signature $sig$ is valid on $m$.
\end{itemize}
\subsection{Learning With Errors}
\label{subsec:lwe}
Given a finite field $\mathbb{F}_q$ and a discrete probability distribution $\mathcal{X}$ over $\mathbb{F}_q$. Let $\bm{s} \in \mathbb{F}_q^v$ be a secret vector, $\bm{A} \in \mathbb{F}_q^{u \times v}$ be a matrix that is chosen uniformly at random and $\bm{e} \in \mathbb{F}_q^u$ be the error vector that is sampled from $\mathcal{X}$. $(v,q,\sigma)$ parameterize an LWE instance, where $\sigma$ is the standard deviation of $\mathcal{X}$. The Learning With Errors (LWE) (search) problem is to find $\bm{s}$, given the pair $(\bm{A}, \bm{b})$, where $\bm{b} = \bm{A} \bm{s} + \bm{e}$. And the LWE decision problem is to distinguish between two uniformly randomly generated pairs.
\cite{regev2009lattices} shows that if the size of $q$ is polynomial in $v$ and $\mathcal{X}$ is a discrete Gaussian distribution on $\mathbb{F}_q$ with standard deviation $\sigma > 2\sqrt{v}$, the LWE decision problem is at least as hard as the LWE search problem and solving the LWE search problem can be reduced to solving the Shortest Vector Problem, a well-known NP-hard problem. In this work, $v=\lambda$, and we use $\mathbb{Z}_p$ as $\mathbb{F}_q$. 
\section{Partial Vector Freezing}
\label{sec:method}
\textbf{Scenario.}
In the $t$-th round of FL, the user set $\mathcal{U}=\{u_1,\ldots,u_n \}$ conduct local model training and submit model updates $\{\bm{x}^{i(t)}\}_{i\in \mathcal{U}} = \{(x^{i(t)}_1;\ldots;x^{i(t)}_m)\}_{i\in \mathcal{U}}$ (\textit{Original Vectors}) to the server $\mathcal{S}$. There might be $\eta~(\le 30\%)$ users that drop out during the aggregation due to network instability or other reasons and $\mathcal{U}'$ is the surviving user set. $\mathcal{S}$ aggregates the model updates to compute $\sum_{i\in \mathcal{U}'} \bm{x}^{i(t)}$ and redistributes the result to all users (\textit{Plain Aggregation}). This iterative process continues until the completion of model training. SAPs can help obtain $\sum_{i\in \mathcal{U}'} \bm{x}^{i(t)}$ while ensuring the privacy of each individual $\bm{x}^{i(t)}$. Similar to other SAPs~\cite{Bonawitz2017PracticalSA, stevens2022efficient}, we define the elements of $\bm{x}^{i(t)}~(i\in [1,n])$ within $\mathbb{Z}_p$ for some large public prime $p$ and assume there is a secure communication channel between each user and $\mathcal{S}$. In this section, our emphasis lies in introducing the computation methodology of \sysname, specifically considering a single-round aggregation process with superscript ``$^{(t)}$'' omitted. To ensure multi-round privacy, employing a specialized user selection mechanism for each round is sufficient~\cite{Liu2023LongTermPA, so2023securing}. For the summary of notations, please refer to \cref{sec:notation}.

\textbf{Threat model}: Corrupt participants endeavor to infer the privacy of honest parties based on the messages they receive, i.e., the \textit{Semi-honest Model}, and can fabricate messages, i.e., the \textit{Active Adversary Model}. We assume malicious users do not exceed one-third of the total users, aligned with PracAgg~\cite{Bonawitz2017PracticalSA} and EffiAgg~\cite{Liu2022EfficientDA}. The lenient security assumptions of \sysname\ allow its easy integration with secure aggregation protocols. When integrating with a SAP, the security assumptions originally employed in the corresponding protocol are adopted and we assume the integrated SAP can \textbf{reliably ensure the privacy of inputs}.

\subsection{Motivation}
Within SAP, every minor operation on an entry accumulates $m$ times, ultimately imposing significant computational burdens on devices. For example, $u_i$ in PracAgg needs to perform the following calculations on each entry $x^{i}_j$ of $\bm{x}^i$ ($b$ and $s$ are user secret keys, PRG is a pseudorandom number generator):
\begin{equation}
    y^i_j=x^i_j  +\operatorname{PRG}(b_i) + \sum_{h \in \mathcal{U}: i<h} \operatorname{PRG}(s_{i, h}) - \sum_{h \in \mathcal{U}: i>h} \operatorname{PRG}(s_{h, i}).
\end{equation}
Based on these observations, we try to reduce the number of entries processed in secure aggregation while ensuring users receive all entries of the aggregated vector. To accomplish this objective, our mindset is to devise a module that can freeze certain entries (\textit{Frozen Entries}) of the user's original vector, and only perform secure aggregation on the other entries (\textit{SecAgg Entries}).  
Upon SAP completion, this module thaws the frozen entries, ensuring that their privacy is well protected throughout the entire process. 

\begin{definition}
Given an invertible $\lambda \times \lambda$ matrix $\bm{A}$ and a segment of original vector $\bm{x}=(x_{1};\ldots;x_{\lambda})$, we define \textit{Incomplete Matrix} $\bm{\check{A}}=\bm{A}_{:\lambda-1, :}$ and \textit{Residual Vector} $\bm{\alpha}=\bm{A}_{\lambda, :}$. The function for finding the solution set of a system of linear equations is:
\begin{equation}
    SLE_{AK}(\bm{Ax})\rightarrow \bm{x},
\end{equation}
where $AK$ denotes the additional knowledge. We use $rank(\cdot)$ represents the rank of a matrix. Since $rank(\bm{A}) = rank(\bm{A,Ax})=\lambda$, the system of linear equations has a unique solution $\bm{x}$. However, due to $rank(\bm{\check{A}}) = rank(\bm{\check{A},\check{A}x})<\lambda$, the system has an infinitude of solutions~\cite{suetin1989linear}, also called an \textit{Under-determined System of Linear Equations}. 
\end{definition}

It can be seen that when $\bm{\check{A}}$ and $\bm{\check{A}x}$ are known, in the absence of knowledge about $\bm{\alpha x}$, $\bm{x}$ presents an infinite set of possibilities, rendering it impossible to determine that specific confidential vector. Motivated by this, we propose \sysname.

\begin{figure}[t!]
\centering
\includegraphics[width=\linewidth]{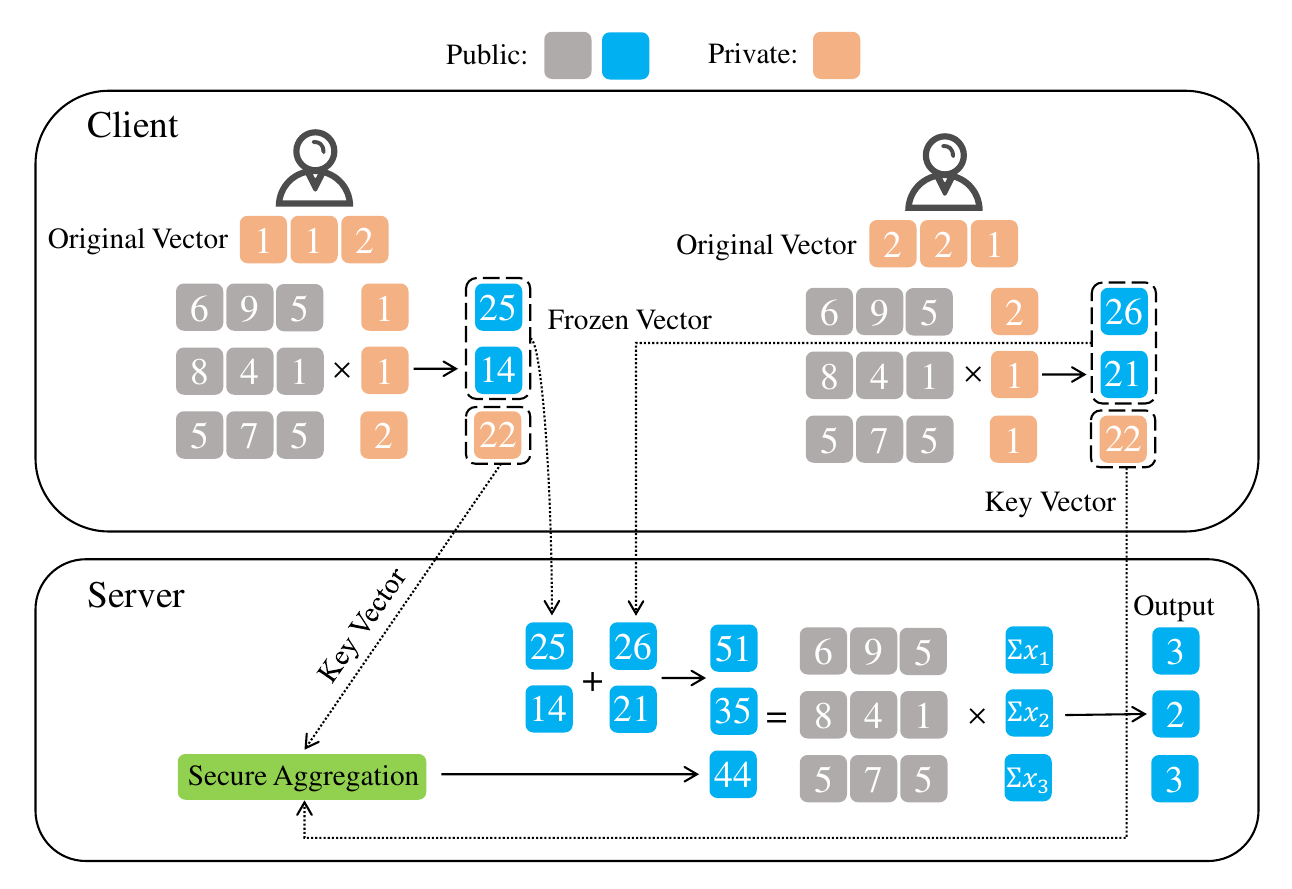}
\caption{Workflow of $\lambda$-SecAgg with Main \sysname.}
\label{fig:workflow}
\end{figure}

\subsection{Main Method}
\label{mainmethod}
In this section, we present the computation process of the Main \sysname\ module during a single aggregation round, depicted by the workflow shown in Figure \ref{fig:workflow}. 

\textbf{Phase 0: Main.Setup($\cdot$)}. Generate an invertible matrix $\bm{A}\in \mathbb{Z}_p^{\lambda \times \lambda}$ randomly, and obtain $\bm{A^{-1}}$, $\bm{\check{A}}$, $\bm{\alpha}$, which are all \textbf{public} parameters. We refer to $\lambda$ as the \textit{Compression Factor}.

\textbf{Phase 1: Main.Freeze($\cdot$)}. For $i\in [1,n]$, randomly pad $\bm{x}^{i}$ to ensure the length of the padded vector is a multiple of $\lambda$. The padded vector is $\bm{x}^{i}_{pad}=(x^{i}_1;x^{i}_2;\ldots;x^{i}_{m'}),$
where $m'=l\lambda$. Then divide $\bm{x}^{i}_{pad}$ into $l$ groups: $\bm{x}^{i}_{pad}=(\bm{d}^{i}_1;\bm{d}^{i}_2; \ldots;\bm{d}^{i}_l)$, where $\bm{d}^{i}_j=(x^{i}_{(j-1) \lambda +1};x^{i}_{(j-1) \lambda +2};\ldots;x^{i}_{j \lambda })$.
Use the incomplete matrix $\bm{\check{A}}$ to compute \textit{Frozen Vector}:
\begin{equation}
    \label{eq:y}
    \begin{aligned}
\bm{y}^{i}& =\left(\bm{y}^{i}_1,\ldots,\bm{y}^{i}_l\right) \\
&=\left(\left(y^{i}_{1},\ldots,y^{i}_{\lambda-1}\right),\ldots,\left(y^{i}_{(l-1)\lambda+1},\ldots,y^{i}_{(l-1)\lambda+\lambda-1}\right)\right)  \\
& =\left(\bm{\check{A}}\bm{d}^{i}_1,\ldots,\bm{\check{A}}\bm{d}^{i}_l\right),
    \end{aligned}
\end{equation}
and use the residual vector to compute \textit{Key Vector}:
\begin{equation}
    \label{eq:k}
        \bm{k}^{i} =\left(k^{i}_1,\ldots,k^{i}_l\right)
        =\left(\bm{\alpha}\bm{d}^{i}_1,\ldots,\bm{\alpha}\bm{d}^{i}_l\right)
         .
\end{equation} 

\textbf{Phase 2: Main.SecAgg($\cdot$)}. Users and $\mathcal{S}$ execute SAP, where the vector to be aggregated of user $i$ is $\bm{k}^{i} =(k^{i}_1,k^{i}_2,\ldots,k^{i}_l) \in \mathbb{Z}_p^l$. We require all users to send their respective $\bm{y}^{i}$ to $\mathcal{S}$ during SAP, thus \textbf{eliminating the need for additional interactions}. Upon completion of SAP, the surviving user set $\mathcal{U}'$ and $\mathcal{S}$ can receive the aggregated result of the key vectors: $\sum_{i\in\mathcal{U}'} \bm{k}^{i}$ or $Enc(\sum_{i\in\mathcal{U}'} \bm{k}^{i})$ (in HE-based SAPs). Then $\mathcal{S}$ computes $\sum_{i\in \mathcal{U}'} \bm{y}^{i}$, immune to the impact of dropout users. 

\textbf{Phase 3: Main.Thaw($\cdot$)}. Thawing can be executed either at the server or user side, without compromising privacy and incurring any additional communication:

(1) \textit{Thawing on the server side}. Based on the acquired $\sum_{i\in\mathcal{U}'} \bm{k}^{i}$ and $\sum_{i\in\mathcal{U}'} \bm{y}^{i}$, $\mathcal{S}$ can derive the aggregated result $\bm{sum}$. The correctness stems from the linearity of $\bm{A}$ and $\bm{\alpha}$:
\begin{equation}
\label{eq:z}
    \resizebox{0.9\columnwidth}{!}{$\begin{aligned}
        \bm{z} & = \left(\left(\sum_{i\in \mathcal{U}'}\bm{y}^{i}_1,\sum_{i\in \mathcal{U}'}k^{i}_1\right),\ldots,\left(\sum_{i\in \mathcal{U}'}\bm{y}^{i}_l,\sum_{i\in \mathcal{U}'}k^{i}_l\right)\right) \\
        & = \left(\left(\sum_{i\in \mathcal{U}'}\bm{\check{A}}\bm{d}^{i}_1,\sum_{i\in \mathcal{U}'}k^{i}_1\right),\ldots,\left(\sum_{i\in \mathcal{U}'}\bm{\check{A}}\bm{d}^{i}_l,\sum_{i\in \mathcal{U}'}k^{i}_l\right)\right) \\
                   & = \left(\sum_{i\in \mathcal{U}'}\bm{A}\bm{d}^{i}_1,\ldots,\sum_{i\in \mathcal{U}'}\bm{A}\bm{d}^{i}_l\right)   \\
                   &= \left(\bm{A}\sum_{i\in \mathcal{U}'}\bm{d}^{i}_1,\ldots,\bm{A}\sum_{i\in \mathcal{U}'}\bm{d}^{i}_l\right)                                                                                                                
    \end{aligned}$}.
\end{equation}
Since $\bm{A^{-1}}$ and $\bm{z}$ are public, $\mathcal{S}$ can thaw frozen vectors and compute the aggregated results of all entries in the original vector by:
\begin{equation}
    \label{eq:lineareqs}
    \begin{aligned}
         & \sum_{i\in \mathcal{U}'}\bm{d}^{i}_1=\bm{A^{-1}}\bm{z}_1=\bm{A^{-1}}\bm{A}\sum_{i\in \mathcal{U}'}(x^{i}_{1};x^{i}_{2};\ldots;x^{i}_{\lambda }),                                 \\
         & \vdots                                                                                                                                                                                                  \\
         & \sum_{i\in \mathcal{U}'}\bm{d}^{i}_l=\bm{A^{-1}}\bm{z}_l=\bm{A^{-1}}\bm{A}\sum_{i\in \mathcal{U}'}(x^{i}_{(l-1) \lambda +1};x^{i}_{(l-1) \lambda +2};\ldots;x^{i}_{l \lambda }).
    \end{aligned}
\end{equation}
At this point, $\mathcal{S}$ completes the thawing phase and obtain:  
\begin{equation}
\begin{aligned}
            \bm{sum} &=\left(\sum_{i\in \mathcal{U}'}\bm{d}^{i}_1;\ldots;\sum_{i\in \mathcal{U}'}\bm{d}^{i}_l\right) = \left(\sum_{i\in \mathcal{U}'}{x}^{i}_1;\ldots;\sum_{i\in \mathcal{U}'}{x}^{i}_{l\lambda}\right) \\
        &= \sum_{i\in \mathcal{U}'}\bm{x}^{i}_{pad}.
\end{aligned}
\end{equation}
Subsequently, it transmits $\bm{sum}$ to all online users. Upon receiving $\bm{sum}$, users can obtain the final aggregated result by removing the padding.

(2) \textit{Thawing on the user side}. In certain SAPs~\cite{aono2017privacy, xu2022non}, the aggregated results remain invisible to $\mathcal{S}$ to ensure the protection of users' intellectual property, among other goals. In such situations, $\mathcal{S}$ cannot perform the thawing. Instead, $\mathcal{S}$ sends $\sum_{i\in \mathcal{U}'}\bm{y}^{i}$ and $Enc(\sum_{i\in\mathcal{U}'} \bm{k}^{i})$ to all surviving users. Users locally decrypt $Enc(\sum_{i\in\mathcal{U}'} \bm{k}^{i})$ and perform the thawing process to obtain $\bm{sum}$. By removing the padding, users attain the final aggregated result.

The privacy of $\bm{x}^i$ in the Main PVF is primarily safeguarded by the hardness of determining a specific solution to an under-determined system of linear equations.

\subsection{Disrupting Variables Extension}
\begin{figure*}[t!]
    \begin{center}
        \resizebox{\linewidth}{!}{\begin{tabular}{|l|}
                \hline
                \multirow{2}* {\centerline{\textbf{\quad \quad \quad \quad \quad \quad \quad \quad \quad \quad $\lambda$-SecAgg Protocol}} }                                                                                                                                                                                           \\
                {}                                                                                                                                                                                                                                                           \\
                \textbf{Participants}: $\mathcal{S}$ and User set $\mathcal{U}=\{u_1,u_2,\ldots,u_n\}$.                                                                                                                                                                      \\
                \textbf{Public Inputs}: $\bm{A}$, $\mu$, $\check{\bm{A}}$, $\bm{\alpha}$, $\lambda$, $\mathbb{Z}_p$, $g$ and $h$. Users' public keys for signatures $\{ sig_i^{pk}\}_{i \in \mathcal{U}}$ and the server's public key for signatures $sig_S^{pk}$.\\
                \textbf{Private Inputs}: Original vectors $\{ \bm{x}^{i(t)}\}_{i \in \mathcal{U}}$ of $t$-th iteration. Users' secret keys for signatures $\{ sig_i^{sk}\}_{i \in \mathcal{U}}$ and the server's secret key for signatures $sig_S^{sk}$.                                                                                                                                                 \\
                \textbf{Outputs}: Surviving user set $\mathcal{U}'$, $\sum_{i\in \mathcal{U}'} \bm{x}^{i(t)}$.                                                                                                                                                           \\
                $~\bullet~$\textbf{Phase 1  Freezing}                                                                                                                                                                                                                        \\
                \quad User $i \in \mathcal{U}$:                                                                                                                                                                                                                              \\
                \quad\quad - pad $\bm{x}^{i(t)}$ randomly and group the entries.                                                                                                                                           \\
                \quad\quad - add noise to $\bm{x}^{i(t)}$ via \cref{eq:noise}.                                                                                                                                           \\
                                \quad\quad - calculate key vector $\bm{k}^{i(t)}$ via \cref{eq:k_sec}.                                                                                                                                                                                     \\
                \quad\quad - calculate frozen vector $\bm{y}^{i(t)}$ via \cref{eq:y}.                             \\
                \quad\quad - obtain $m_1^i=$ $\bm{y}^{i(t)}$, send $\sigma_{1}^i \rightarrow DS.sign( {{sig}_{i}^{sk},m_1^i} )$ to $\mathcal{S}$.                                                                                                                                      \\
                $~\bullet~$\textbf{Phase 2  SecAgg}                                                                                                                                                                                                                          \\
                \quad $\mathcal{S}$ and Users:                                                                                                                                                                                                                               \\
                \quad\quad - execute \textbf{SAP} for $\{\bm{k}^{i(t)}\}_{i\in \mathcal{U}}$. \\
                \quad\quad - all participants receive $\mathcal{U}'$ and $\sum_{i\in\mathcal{U}'}\bm{k}^{i(t)}$ (or $Enc(\sum_{i\in\mathcal{U}'}\bm{k}^{i(t)})$).                               \\

                \quad $\mathcal{S}$:                                                                                                                                                                                                                                         \\
                \quad \quad - if $DS.verify( {\sigma_{1}^i, {sig}_{i}^{pk},m_1^i} )\rightarrow False$, abort. Otherwise, calculate $\sum_{i\in \mathcal{U}'}\bm{y}^{i(t)}$.                                                                                        \\
                $~\bullet~$\textbf{Phase 3  Thawing}                                                                                                                                                                                                                         \\
                \quad \textit{Thawing on the server side}                                                                                                                                                                                                                    \\
                \quad $\mathcal{S}$:                                                                                                                                                                                                                                         \\
                                \quad\quad - calculate $\bm{sum}=\sum_{i\in \mathcal{U}'}\bm{x}^{i(t)}_{pad}$ via \cref{eq:lineareqs} and send $\bm{sum}$ and $\sigma_{3} \rightarrow DS.sign( {{sig}_{S}^{sk},\bm{sum}} )$ to $i\in \mathcal{U}'$.                                                                                      \\
                \quad User $i \in \mathcal{U}'$:                                                                                                                                                                                                                             \\
                \quad\quad - receive $\sum_{i\in \mathcal{U}'}\bm{x}^{i(t)}$, if $DS.verify( {\sigma_{3}, {sig}_{S}^{pk},\bm{sum}} )\rightarrow False$, abort. Otherwise, unpad and output.                                                                                                                                                                                  \\
                \quad \textit{Thawing on the user side}                                                                                                                                                                                                                      \\
                \quad $\mathcal{S}$:                                                                                                                                                                                                                                         \\
                \quad\quad - send $\sum_{i\in \mathcal{U}'}\bm{y}^{i(t)}$, $Enc(\sum_{i\in\mathcal{U}'}\bm{k}^{i(t)})$ and $\sigma_{3} \rightarrow DS.sign( {{sig}_{S}^{sk},\sum_{i\in \mathcal{U}'}\bm{y}^{i(t)}}|| Enc(\sum_{i\in\mathcal{U}'}\bm{k}^{i(t)}))$ to $i \in \mathcal{U}'$.                                                                                                                   \\
                \quad User $i \in \mathcal{U}'$:                                                                                                                                                                                                                             \\
                \quad\quad - if $DS.verify( {\sigma_{3}, {sig}_{S}^{pk},\sum_{i\in \mathcal{U}'}\bm{y}^{i(t)}}|| Enc(\sum_{i\in\mathcal{U}'}\bm{k}^{i(t)} )\rightarrow False$, abort.                                                                                                    \\
                \quad\quad - decrypt $Enc(\sum_{i\in\mathcal{U}'}\bm{k}^{i(t)})$.                                                                                                                                                                                        \\
                \quad\quad - calculate $\bm{sum}=\sum_{i\in \mathcal{U}'}\bm{x}^{i(t)}_{pad}$ via \cref{eq:lineareqs}, unpad and output.                                                                                                              \\
                {}                                                                                                                                                                                                                                                           \\
                \hline
            \end{tabular}}
    \end{center}
    \caption{The pipline of $\lambda$-SecAgg for one aggregation.}
    \label{fig:protocol}
\end{figure*}

\label{extension:DVE}
While the server cannot obtain any individual element of $\bm{x}^i$ within Main PVF, it still obtains certain \textbf{linear relationships} involving the private entries. In this section, we present improvements to the Main PVF to ensure that the server cannot obtain any information about $\bm{x}^i$ from $\bm{y}^i$.


Similar to many hybrid schemes combining mask and DP~\cite{Bonawitz2017PracticalSA,stevens2022efficient,liu2022sash} (or encryption and DP~\cite{wang2021protecting}), DVE adds noise to $\bm{x}^i$ before aggregation to enhance privacy by: 
\begin{equation}
\label{eq:noise}
    \bm{x}^i = \bm{x}^i + \bm{e},
\end{equation}
where the Gaussian noise $\bm{e} \sim \mathcal{N}(0, \bm{\Sigma} _ e=\bm{A}^{-1}\bm{D}\bm{A}^{-T})$ and $\bm{\Sigma} _ e$ is the covariance matrix of $\bm{e}$ and $\bm{D}$ is a diagonal matrix. 
Then $\bm{y}^i=\check{\bm{A}}(\bm{x}^i+ \bm{e})=\check{\bm{A}}\bm{x}^i+ \bm{e'}$. $\bm{e}'$ follows a multivariate normal distribution, with its mean and covariance matrix given by:

\begin{equation}
\label{eq:eplus-mean}
\mathbb{E}[\bm{e}'] = \mathbb{E}[\bm{A}\bm{e}] = \bm{A}\mathbb{E}[\bm{e}] = \bm{A} \cdot 0 = 0,
\end{equation}
\begin{equation}
\label{eq:eplus-var}
\text{Cov}(\bm{e}') = \bm{A}\text{Cov}(\bm{e}) \bm{A}^T = \bm{A} \Sigma _ e \bm{A}^T= \bm{A} \bm{A}^{-1}\bm{D}\bm{A}^{-T} \bm{A}^T=\bm{D},
\end{equation}
which means each component of $\bm{e}'$ \textbf{independently} follows a Gaussian distribution. Each $e$ follows a different Gaussian distribution does not affect the hardness of the LWE problem. For instance, consider LWE samples $<\bm{a}, \bm{ax} + e_1>$ and $<\bm{a}, \bm{ax} + e_2>$. They can be viewed as belonging to different LWE instances, which does not compromise security.

Therefore, given a uniformly random vector $\bm{w}^i$, Sec. \ref{subsec:lwe} ensures that $(\check{\bm{A}}, \bm{y}^i)$ and $(\check{\bm{A}}, \bm{w}^i)$ are indistinguishable, which guarantees \textbf{$\mathcal{S}$ does not obtain private information from honest users through frozen vectors.}
For clarity and conciseness, we do not differentiate the symbols of the original vector before and after adding noise. 

The standard deviation of the noise added to $x_j$ is $\sigma _ j = (\bm{A}^{-1}\bm{D}\bm{A}^{-T}) _ {jj} = \sigma _ j' \sum _ {k=1}^\lambda (\bm{A}^{-1}) _ {jk}^2$, where $\sigma _ j'$ is the $j$-th diagonal element of the $\bm{D}$. Since $p$ is a prime number, i.e., $\gcd(\sum _ {k=1}^\lambda (\bm{A}^{-1}) _ {jk}^2, p) = 1$, the map $\sigma _ j' \mapsto \sigma _ j' \sum _ {k=1}^\lambda (\bm{A}^{-1}) _ {jk}^2 \mod p$ is a bijective mapping from $\mathbb{Z} _ p$ to $\mathbb{Z} _ p$. $\sigma _ j'$ must satisfy $\sigma _ j' > 2\sqrt{\lambda}$, which means that the range of invalid values for $\sigma _ j'$ is extremely small. Therefore, $\sigma_j' \sum_{k=1}^\lambda (\bm{A}^{-1})_{jk}^2$ can basically take all values of $\mathbb{Z} _ p$.
It is evident that we can always choose a $\sigma _ j'$ such that
\begin{equation}
\label{eq:sigma}
\sigma _ j' \sum _ {k=1}^\lambda (\bm{A}^{-1}) _ {jk}^2 \mod p < \epsilon,
\end{equation}
where $\epsilon$ is a small number representing the maximum noise that can be added to $x _ j$.

For example, in Figure \ref{fig:dp}, the range of invalid values for $\sigma _ i'$ is $[0, 2\sqrt{1000} = 63]$, meaning that there are only 64 specific values in $\mathbb{Z} _ p$ that $\sigma _ i' \sum _ {k=1}^\lambda (\bm{A}^{-1}) _ {ik}^2$ cannot take. Since $\epsilon = 8783$ in Figure \ref{fig:dp} still ensures that the added noise is negligible and there are only 64 values that cannot be taken, $\sigma _ i' \sum _ {k=1}^\lambda (\bm{A}^{-1}) _ {ik}^2 \mod p$ can always take values within the range $[0, 8783]$. Thus, we can always construct a $\bm{D}$ that not only ensures the noise added to $\bm{x}$ is negligible but also guarantees that the components of $\bm{e}'$ in the LWE instance are independent of each other and $\sigma _ i' > 2\sqrt{\lambda}$.

\subsection{Integrating \sysname\ with Different SAPs}
\label{portability}
To resist \textbf{model inconsistency attacks}, appending the hash of the received model to the pseudorandom generator seed is sufficient, without incurring additional overhead~\cite{Ma2023FlamingoMS}. And utilizing an innovative user selection mechanism~\cite{so2023securing,Liu2023LongTermPA} is able to achieve \textbf{multi-round privacy}. To resist \textbf{aggregation falsification attacks}, verifiable protocols~\cite{Hahn2023VerSAVS} are able to verify the aggregation results through commitments sent by $\mathcal{S}$, and we provide \textit{Result Verification Extension} in Appendix~\ref{ex:rve} to enable PVF to integrate with such SAP.  
The pipeline of $\lambda$-SecAgg is shown in Figure \ref{fig:protocol}. We omit the encryption of messages before transmission and the decryption after reception by each participant. If any error occurs during encryption or decryption process, the protocol will be immediately terminated. For the detailed portability analysis, please refer to Appendix~\ref{app:port}.

\section{Security Analysis}
\label{sec:analysis}
Evidently, the information that adversaries can obtain about an honest participant only includes $\bm{sum}$ and under-determined systems of linear equations ($\bm{y}^{i}$). Randomly generating $\bm{A}$ and performing certain \textbf{pre-checks} (as \cref{security_single}), the under-determined systems of linear equations can effectively preserve the privacy of each entry, which is also adopted by \cite{Liu2023LongTermPA}. And in \cref{security_main}, we demonstrate adversaries cannot access $\bm{x}^{i}$ throughout $\lambda$-SecAgg process.
\subsection{Privacy of Each Element}

\label{security_single}
In cases of improper selection of $\bm{A}$, $\mathcal{S}$ can access the privacy of some specific elements within an original vector (if without DVE), as illustrated in Example~\ref{example:sec}. 
In the implementation, we initially generate $\bm{A}$ randomly, and we transform $\check{\bm{A}}$ into \textit{Reduced Row Echelon Form} and verify that no element can be deduced. This ensures privacy of every element. Unless specified otherwise, all subsequent references to $\bm{A}$ in the following text are designed to guarantee each element privacy.

\begin{empheqboxed}
\begin{example}
\label{example:sec}
Consider a $3\times 3$ matrix:
\begin{equation}
    \bm{A}=\begin{bmatrix}
        1 & 2 & 3 \\
        1 & 3 & 3 \\
        1 & 2 & 4
    \end{bmatrix},
\end{equation}
which is an invertible matrix. The corresponding incomplete matrix is:
\begin{equation}
    \check{\bm{A}}=\begin{bmatrix}
        1 & 2 & 3 \\
        1 & 3 & 3
    \end{bmatrix}.
\end{equation}
Assume the original vector is $\bm{x}=(x_1;x_2;x_3)=(1;2;3)$, and the frozen vector is $\check{\bm{A}}\bm{x}=(14,16)$.
$\mathcal{S}$ obtains the under-determined system of linear equations as follows:
\begin{equation}
    \left\{\begin{array}{l}
        x_1+2x_2+3x_3=14 \quad(i)\\
        x_1+3x_2+3x_3=16 \quad(ii)
    \end{array}\right.
\end{equation}
\end{example}
\end{empheqboxed}
While $\mathcal{S}$ cannot obtain the complete $\bm{x}$, it can deduce $x_2=2$ by $(ii)-(i)$. Unlike $\mathcal{S}$ obtaining prior knowledge of elements through attacks as discussed in Section~\ref{securityex}, here, $\mathcal{S}$ deduces $x_2=2$ through the computation process within \sysname. However, since $\bm{A}$ is public, any maliciously constructed $\bm{A}$ can be easily detected by honest users.

\subsection{Security Analysis of Entire Vectors}
\label{security_main}
First and foremost, it is evident that in PVF, the user only transmits $\mathbf{y}^i$ to the server, and the server only sends $\sum \mathbf{y}^i$ back to the users after SAP ends. Notably, $\mathbf{k}^i$ and $\sum \mathbf{k}^i$ are transmitted through the SAP, independent of PVF. In the active adversary model, we obviously cannot guarantee the correctness of the aggregation result because the malicious server can arbitrarily modify the result. However, we can guarantee the privacy of honest users’ inputs. We provide a detailed explanation of the active attacks that malicious participants can launch within PVF and how PVF leverages the cryptographic primitives to defend against them.
\begin{itemize}[leftmargin=*]
    \item Forging fake users to participate in PVF. This type of attack, also known as a \textit{Sybil Attack}, involves fake users reporting received information to the server. Such attacks primarily target scenarios where users share secret keys among themselves but keep the keys secret from the server, like PPDL. Alternatively, an attacker may attempt to forge a large number of fake users (more than $\frac{1}{3}|\mathcal{U}|$) to reconstruct users’ private keys in the secret-sharing scheme. Since PVF does not involve information that is kept secret from the server but shared among all users, and consistent with the assumption in PracAgg that the number of malicious users does not exceed $\frac{1}{3}|\mathcal{U}|$, PVF is resistant to this type of attack.
    \item Attempting to forge or tamper with honest users' messages. Such attacks may occur in PVF in the following situations: malicious participants forging or tampering with an honest user's $\mathbf{y}^i$. This can be avoided by the digital signature $\sigma_{1}^{i}$ employed in PVF. Similarly, malicious participants may attempt to forge or tamper with $\sum \mathbf{y}^i$ sent by the server, which is prevented by the use of $\sigma_{3}$. 
    \item Sending malformed messages. In PVF, such attacks include malicious users sending malformed ciphertexts of $\mathbf{y}^i$ or the malicious server sending malformed ciphertexts of $\sum \mathbf{y}^i$. Such attacks are prevented by the IND-CPA and IND-CTXT security of the symmetric authenticated encryption used in PVF. If decryption fails, the protocol is immediately terminated. 
    \item Intercepting and stealing private information. Malicious adversaries may intercept messages sent by honest users to extract private information. This is effectively avoided by the symmetric authenticated encryption employed in PVF.
\end{itemize}
The use of symmetric authenticated encryption and digital signatures to ensure privacy under the active adversary model is a relatively mature application in the field of secure aggregation, and our design follows these existing works. 
Then we present the following lemmas:
\begin{lemma}[Privacy during the freezing phase]
    \label{ana:lemma1}
     Fix $p$, $\mathcal{U}$, $m$, $\lambda$, $\bm{A}$ and a private vector $\bm{x}^{i} = (\bm{d}^{i}_1,\ldots,\bm{d}^{i}_{ \lceil\frac{m}{\lambda} \rceil})$ (with noise) of an honest user $i \in \mathcal{U}$. For any probabilistic polynomial-time (PPT) adversary $M$ who is given $\{\bm{y}^{i}\}_{i \in \mathcal{U}}$ and $\mathcal{C}$, the advantage of M to obtain any unbroken element $a_j$ is defined as:
    \begin{equation}
        Adv_{M}^y(\lambda):=Pr[SLE_{\mathcal{C}}(\check{\bm{A}}\bm{d}^{i}_j)\rightarrow a_j]_{j\in [1,\lceil\frac{m}{\lambda}   \rceil]}.
    \end{equation}
There exists a negligible function $\varepsilon$ such that $Adv^y_{M}(\lambda) \le \varepsilon$.
\end{lemma}
\begin{remark}
    The adversary can obtain $l(\lambda-1)$ independent equations, with $l\lambda$ variables. Hence there are infinite solutions, and the probability of $M$ determining the unique $(\bm{x}^i)$ is $\frac{1}{\infty}$. 
\end{remark} 

\begin{lemma}[Privacy during the thawing phase]
    \label{ana:lemma2}
    Fix $p$, $\mathcal{U}$, $m$, $\lambda$, $\bm{A}$ and the sum of private vectors $\sum_{i\in\mathcal{U}'} \bm{x}^{i}$. For any PPT adversary $M$ who is given $\{\bm{y}^{i}\}_{i \in \mathcal{U}'}$, $\mathcal{C}$ and $\sum_{i\in \mathcal{U}'}\bm{k}^{i}$, the advantage of M to obtain any unbroken element $a_j$ is defined as:
    \begin{equation}
        Adv_{M}^{y,k}(\lambda):=Pr[SLE_{\mathcal{C}}(\bm{A}\sum_{j\in \mathcal{U}'}\bm{d}^{i}_j)\rightarrow a_j]_{j\in [1,\lceil\frac{m}{\lambda}   \rceil]}.
    \end{equation}
    There exists a negligible function $\varepsilon$ such that $Adv_{M}^{y,k}(\lambda) \le \varepsilon$.
\end{lemma}

\begin{remark}
    $\mathcal{S}$ can obtain $\sum_{j\in \mathcal{U}'}\bm{d}^{i}_j$ by \cref{eq:lineareqs}. For any individual $\bm{d}^{i}_j$ (with noise), the information known to $M$ is $\sum_{j\in \mathcal{U}'}\bm{d}^{i}_j$ and $\check{\bm{A}}\bm{d}^{i}_j$. So there are still infinite solutions, and $Adv_{M}^{y,k}(\lambda)=\frac{1}{\infty}$.
\end{remark}

Protocol security requires that adversaries can not obtain private information of any \textbf{individual} honest participant. The view of a participant consists of its internal information and received messages. Given a SAP that can maintain security in the active adversary setting, \cref{theorem} guarantees its security when integrated with PVF. \cref{theorem} demonstrates the information revealed during a real execution is \textbf{indistinguishable} from that obtained through a random simulation. The proof of \cref{theorem} is is carried out in a \textit{Random Oracle model}. In this model, we define a trapdoor function to inform $SIM$ of the \textbf{sum} of existing honest users' private information. During one execution of the protocol, $SIM$ can only access it once to obtain necessary information. Let $\mathcal{C}$ denote the set of malicious participants, which is a subset of $\mathcal{U} \cup \{\mathcal{S}\}$. The ideal function is defined as follows:
\begin{equation}
    {Ideal}_{{\{ \bm{x}^{i}\}}_{i \in \mathcal{U}\backslash \mathcal{C}}}(L) = \left\{
    \begin{aligned}
        \sum_{i \in L}\bm{x}^{i} & ,~L \subseteq (\mathcal{U}\backslash \mathcal{C})~and~\left | L \right | >\left \lceil \frac{n}{3} \right \rceil   \\
        \bot                            & ,~otherwise                                                         \\
    \end{aligned}
    \right. .
\end{equation}
\begin{theorem} [Security against malicious participants]
\label{theorem}
Let $\mathit{REAL}_{\mathcal{C}}^{\mathcal{U},\lambda}( \{\bm{x}^{i}\}_{i\in \mathcal{U}},\mathcal{U}')$ denote a random variable representing the joint view of adversaries in an actual protocol execution and $\mathit{SIM}_{\mathcal{C}}^{\mathcal{U},\lambda}( \{\bm{x}^{i}\}_{i\in \mathcal{U}},\mathcal{U}')$ denote the joint view of adversaries in a simulated protocol execution. For all $\lambda>2, \mathcal{U},\bm{x}^{i}_{i\in \mathcal{U}},\mathcal{U}'$, $\mathcal{C} \subseteq \mathcal{U} \cup \{ \mathcal{S}\}$ and SAP that can ensure privacy in the active adversary setting, there exists a PPT simulator $SIM$ such that:
\begin{equation}
    \mathit{SIM}_{\mathcal{C}}^{\mathcal{U},\lambda}( \{\bm{x}^{i}\}_{i\in \mathcal{U}},\mathcal{U}')
    \equiv \mathit{REAL}_{\mathcal{C}}^{\mathcal{U},\lambda}( \{\bm{x}^{i}\}_{i\in \mathcal{U}},\mathcal{U}'),
\end{equation}
where ``$\equiv$'' denotes the distributions are identical.
\end{theorem}

We use a standard hybrid argument to prove the theorem. 
\begin{proof}
We define a sequence of hybrid distributions ${H}_{0},{H}_{1},\ldots$ to denote a series of modifications to $REAL$, which can finally get $SIM$. We prove $SIM$ and $REAL$ are indistinguishable by proving two adjacent hybrids are indistinguishable.
\begin{enumerate}[leftmargin=*]
    \item[$H_0$] In this hybrid, $SIM$ is exactly the same as $REAL$.
    \item[$H_1$] This hybrid is distributed similarly to the previous one, except for the following modifications. $SIM$ obtains $\sum_{i \in \mathcal{U}' \backslash \mathcal{C}}\bm{x}^{i}$ by calling ${Ideal}_{{\{ \bm{x}^{i}\}}_{i \in \mathcal{U}\backslash \mathcal{C}}}( {\mathcal{U}' \backslash \mathcal{C}} )$. $SIM$ aborts if there is an illegal request. We replace the ciphertexts of $\{\bm{y}^{i}\}_{i\in\mathcal{U}}$ with the ciphertexts of uniformly random vectors $\{\bm{w}^{i}\}_{i\in\mathcal{U}}$ that satisfy $\sum_{i \in \mathcal{U}' \backslash \mathcal{C}}\bm{w}^{i} = \sum_{i \in \mathcal{U}' \backslash \mathcal{C}}\bm{y}^{i}$. $\sum_{i \in \mathcal{U}' \backslash \mathcal{C}}\bm{y}^{i}$ can be can be computed from \cref{eq:y} based on $\sum_{i \in \mathcal{U}' \backslash \mathcal{C}}\bm{x}^{i}$. The IND-CPA and IND-CTXT security of symmetric authenticated encryption guarantees the distribution of this hybrid is indistinguishable from the previous one.
    \item[$H_{2}$] This hybrid is distributed exactly as the previous one, except $SIM$ aborts if there is an invalid signature ($\sigma_{1}^{i}$, $\sigma_{2}^{i}$ or $\sigma_{3}$). The UF-CMA security of the digital signature scheme can ensure $\mathcal{C}$ cannot forge any valid signature of an honest user, so the distribution of this hybrid is indistinguishable from the previous one.
    \item[$H_{3}$] This hybrid is distributed similarly to the previous one, except for the following modifications. \textbf{Firstly}, according to the security analysis process of the integrated SAP, we replace the corresponding messages conveyed in SAP with random strings of equal length. \textbf{Secondly}, we replace the frozen vectors $\{\bm{y}^{i}\}_{i\in\mathcal{U}}$ received by $\mathcal{S}$ with uniformly random vectors $\{\bm{w}^{i}\}_{i\in\mathcal{U}}$. $\bm{x}^i$ is added with noise $(\check{\bm{A}}\bm{e})$ through \cref{eq:noise}. Therefore, the hardness of LWE decision problem in \cref{subsec:lwe} ensures that $(\check{\bm{A}}, \bm{y}^i)$ and $(\check{\bm{A}}, \bm{w}^i)$ are indistinguishable, which guarantees $\mathcal{S}$ does not obtain private information from honest users through frozen vectors. Therefore, the security of SAP and the lemmas ensure the distribution of this hybrid is indistinguishable from the previous one.
\end{enumerate}
Therefore, the distribution of $SIM$ which is the same as $H_3$ is indistinguishable from $REAL$. $SIM$ does not depend on the inputs of honest parties, and $\mathcal{C}$ can only learn about the sum of original vectors. If too many users drop out, SAP will abort and still guarantee the above conclusion. Clearly, the security of $\lambda$-SecAgg still holds in the semi-honest setting.
\end{proof}

\section{Evaluation}
\label{sec:evaluation}


\subsection{Theoretical Complexity Analysis}

\textbf{Communication. }PVF does not increase interaction-related overhead. The vectors sent from each user are $\bm{y}^{i}\in \mathbb{Z}_p^{(\lambda-1)\left \lceil \frac{m}{\lambda} \right \rceil}$ and $\bm{k}^{i}\in \mathbb{Z}_p^{\left \lceil \frac{m}{\lambda} \right \rceil}$, which contain the same number of entries as $\bm{x}^{i}_{pad}\in \mathbb{Z}_p^{m'}$. Hence, the theoretical communication complexity of SAP remains unchanged. 

\textbf{Computation. }The additional computation operations required by \sysname\ involve conducting $ \lceil \frac{m}{\lambda} \rceil $ matrix-vector multiplications by both users and $\mathcal{S}$, incurring a computation cost of $O(\lambda m)$. 

The theoretical complexity for one aggregation of users and $\mathcal{S}$, before and after \sysname\ integration, is summarized in Table~\ref{tab:theocom}. For the method to \textbf{avoid padding}, please refer to \cref{app:padding}.
\label{app:theo}
\cref{tab:theocom} indicates for some SAPs, the theoretical computation complexity increases after integrating \sysname. This is attributed to the transformation required for the entire original vector within \sysname. However, intuitively, the computation time for secure aggregation per entry (e.g., homomorphic encryption, PRG expansions, or modular exponentiations) tends to be \textbf{significantly greater compared to the computation time per entry in linear transformations}. This suggests that \sysname\ still manages to compress computation overhead, which is evident in the experiment outcomes. 

In practice, the choice of $\lambda$ mainly considers: (i) Security requirements in DVE (see \cref{subsec:lwe}). (ii) SAP. For schemes with more masking-related overhead, a larger $\lambda$ performs better. For schemes with more interaction-related overhead, a smaller $\lambda$ performs better (as in \cref{comana}). (iii) $m$. For larger $m$, masking-related overhead is greater, so a larger $\lambda$ performs better.

\begin{table*}[!t]
	\begin{center}
		\caption{Theoretical complexity of SAP without and with \sysname\ of single-round aggregation in the semi-honest setting. $\bm{O(\cdot)}(\cdot)$ in the rightmost column indicates the communication complexity pertains to users, followed by the number of interactions between users and the server. ``P. of G. M.'' stands for privacy of the global model. In the real world, $p \gg m \gg n $. $\uparrow$ indicates the integration of \sysname\ into the protocol would increase its theoretical complexity.}
		\resizebox{\linewidth}{!}{\begin{tabular}{ccccccc}
				\toprule
				\multirow{2}{*}{SAP}                  & \multicolumn{2}{c}{Each user} & \multicolumn{2}{c}{Server}                     & P. of                  & \multirow{2}{*}{Communi. (Inter.)}                                                           \\
				\cmidrule[0.5pt](lr){2-3} \cmidrule[0.5pt](lr){4-5}
				{}                                    & w/o \sysname                  & w/ \sysname                                         & w/o \sysname                 & w/ \sysname                                             & G. M.                                  \\
				\midrule[0.5pt]
				PPDL            & $\bm{O(m)}$              & $O(\lambda m)\uparrow$                         & $O(mn)$                 & $O(\frac{1}{\lambda}mn+ \lambda m)$                 & \Checkmark & $\bm{O(m)}$ \textbf{(1)} \\
				EPPFL         & $\bm{O(m)}$                 & $O(\lambda m)\uparrow$                       & $O(mn)$             & $O(\frac{1}{\lambda}mn+\lambda m)$                       & $\times$   & $O(m)$ (2)               \\
				NIVP-DS               & $\bm{O(m)}$              & $O(\lambda m)\uparrow$                         & $O(mn)$                 & $O(\frac{1}{\lambda}mn+ \lambda m)$                 & \Checkmark & $\bm{O(m)}$ \textbf{(1)} \\
				PracAgg & $O(mn+n^2)$              & $O(\frac{1}{\lambda}mn+ \lambda m+n^2)$        & $O(mn^2)$               & $ O(\frac{1}{\lambda}mn^2+\lambda m)$              & $\times$   & $O(m+n)$ (4)             \\
				PracAgg+       & $O(mlogn + log^2n)$      & $O(\frac{1}{\lambda}mlogn +\lambda m+ log^2n)$ & $O(mnlog n + nlog^2 n)$ & $O(\frac{1}{\lambda}mnlog n +\lambda m+ nlog^2 n)$ & $\times$   & $O(m+logn)$ (4)          \\
				EffiAgg     & $O(m+n^2)$               & $O(\lambda m+n^2)\uparrow$                     & $O(m\sqrt{p}+n)$        & $O(\frac{1}{\lambda}m\sqrt{p}+\lambda m+n)$        & $\times$   & $O(m+n)$ (4)             \\
				LPPFedL    & $O(m+n^2)$               & $O(\lambda m+n^2)\uparrow$                     & $\bm{O(m+n)}$           & $O(\lambda m+n)\uparrow$                           & $\times$   & $O(m+n)$ (4)             \\
				\bottomrule
			\end{tabular}}
		\label{tab:theocom}
	\end{center}
\end{table*}

\subsection{Experimental Settings}
\textbf{Baselines}. We include $7$ baselines that encompass $5$ types of \textbf{masking (encryption) schemes}:
\begin{itemize}[leftmargin=*]
    \item \textit{PPDL}~\cite{aono2017privacy}, utilizing \textit{Single-private-key HE} (\textbf{Type 1}), safeguards the global model, all while necessitating only a single interaction round between the users and $\mathcal{S}$.
    \item \textit{EPPFL}~\cite{Li2022EfficientPF}, relying on \textit{Multi-private-key HE} (\textbf{Type 2}), requires two interaction round.
    \item \textit{NIVP-DS}~\cite{xu2022non}, based on \textit{SMPC} (\textbf{Type 3}), requires two non-colluding servers, while involving only a single interaction round.
    \item \textit{PracAgg}~\cite{Bonawitz2017PracticalSA}, based on \textit{Pair-wise Masking} (\textbf{Type 4}), is widely regarded as the state of the art, involving multiple interactions rounds.
    \item \textit{PracAgg+}~\cite{Bell2020SecureSA}, a type-4 scheme, exhibits a more efficient communication structure.
    \item \textit{EffiAgg}~\cite{Liu2022EfficientDA}, based on \textit{Non-pair-wise Masking} (\textbf{Type 5}), requires server computation of discrete logarithms alongside multiple interactions between users and $\mathcal{S}$.
    \item \textit{LPPFedL}~\cite{Wei2023LightweightFL}, based on non-pair-wise masking, necessitates users to transmit multiple high-dimensional vectors along with multiple interactions.
\end{itemize}
To underscore the focal point of our work, our attention is exclusively directed towards secure-aggregation-related operations within the baselines, omitting other parts like ``weight aggregation'' in EPPFL. And masking schemes of other protocols that reduce \textbf{interaction-related} can be classified into the above five types, like Flamingo~\cite{Ma2023FlamingoMS} (type 4), LTPA~\cite{Liu2023LongTermPA} (type 4).

\textbf{Experimental settings}.
We run on a Linux workstation with 32GB of RAM and an AMD Ryzen 5 5600G. We use 1 NVIDIA GeForce RTX 3090 GPU only for model training, excluding the aggregation process. In all experimental settings, the space of the elements in the input vectors is 32-bit. 
Same with PracAgg, our main experiments are conducted under a semi-honest setting, with a specific focus on assessing the speed enhancement brought by \sysname. We disregard operations such as digital signatures and public key infrastructure in the active adversary setting, which do not impact the asymptotics of the results~\cite{Bonawitz2017PracticalSA}. 

\textbf{Evaluation Metrics.} We use $\textit{Improvement Factor}=\frac{\text{Time of (un)masking w/o PVF}}{\text{Time of (un)masking w/ PVF}}$ (speedup) to describe the efficacy of PVF. When evaluating communication overhead, we track costs for a single user, with the costs of $\mathcal{S}$ being $n$ times that of each user, a convention widely adopted in prior works~\cite{Liu2022EfficientDA}. The experimental results provided are the average outcomes from 5 repeated executions.

\textbf{Implementation Details}
We implement the baselines using Python. Specifically, we utilize \textit{AES-GCM} with 128-bit keys for the symmetric authenticated encryption, standard $(t,n)$ \textit{Shamir Secret Sharing}~\cite{shamir1979share}, \textit{AES} in counter mode for the pseudorandom generator, \textit{SHA-256} hash to implement a homomorphic pseudorandom generator for EffiAgg, \textit{Sympy} library~\cite{Meurer_SymPy_symbolic_computing_2017} to compute discrete logarithms for EffiAgg, and \textit{Paillier Encryption} with 1024-bit keys for PPDL. In \cref{extension:DVE}, experimental settings of the image classification task is the same as \cref{fig:end2end}, and for the movie recommendation task, the dataset is split using a \textit{Leave-One-Out}~\cite{he2017neural} approach for training and testing, where users with fewer than 10 records are excluded, using \textit{Hit Ratio} (HR) as metrics to assess the performance of the recommendations, with higher values indicating superior effectiveness. 

In \cref{fig:end2end}, we use three widely discussed SAPs: PracAgg+, PracAgg, and PPDL. NIVP-DS falls behind due to the requirement of two non-colluding servers, while LPPFedL demands increased communication overhead.
Remarkably, PPDL, previously considered impractical due to its use of HE, has its computational overhead significantly mitigated by integrating \sysname. And its advantages of operating with a single server, single-interaction communication, and preserving the global model make it more appealing.

\subsection{Performance Comparison}
\label{comana}

\textbf{Effectiveness of compressing costs. }
Results in \cref{tab:compari} indicate that integrating \sysname\ can yield a speedup ranging from $\bm{70\times}$ to $\bm{99.5\times}$ for the majority of baselines when $\lambda=100$. 
The inability to achieve a $\lambda \times$ speedup is attributed to interaction-related overhead within SAP such as key agreements and secret sharing. 
Overall, the gain from \sysname\ is the weakest for LPPFedL. This is because LPPFedL increases user communication overhead to achieve highly lightweight masking and unmasking. Consequently, interaction-related overhead constitutes a more significant portion of its computation time. As the original vector length increases from $100k$ to $500k$, the proportion of time spent on masking or encryption rises, so \sysname\ exhibits a more noticeable acceleration effect for these SAPs in general.
For PPDL and EPPFL (HE-based SAPs), only $\frac{1}{\lambda}$ of the original vector is transformed into \textbf{ciphertext}. Therefore, \sysname\ brings them communication improvements of about $32.3\times$ and $14.6\times$, respectively.

\textbf{End-to-end comparison. }
In Figure~\ref{fig:end2end}, the model employed is LeNet~\cite{lecun1998gradient}, consisting of $61,706$ parameters. The dataset is MNIST~\cite{deng2012mnist} with each user possessing local data for only two labels. Across different aggregation methods, the model eventually achieves commendable training outcomes, and the integration of \sysname\ showcases a significant acceleration effect, greatly reducing the training speed gap between using SAPs and plain aggregation.
\begin{figure}[!h]
    \centering
    \includegraphics[scale=0.6]{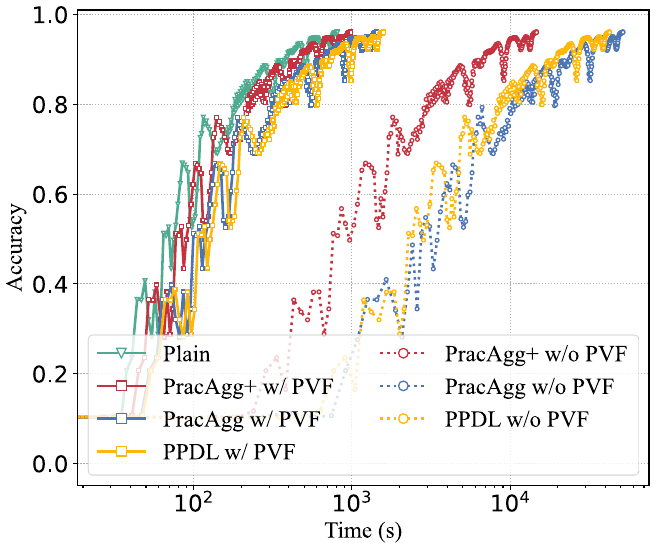}
    \caption{Comparison of various aggregation methods, with $n=100,\eta=5\%,\lambda=100$.}
    \label{fig:end2end}
\end{figure}
In Figure~\ref{fig:compression}, consistent with FetchSGD~\cite{rothchild2020fetchsgd}, we use the ResNet9 (with 6.5M parameters), CIFAR10 dataset and PracAgg, with the optimal setting for FetchSGD and $k=50,000$ for both FetchSGD and Top-k\cite{lu2023top}. It shows our method takes the \textbf{least} time and achieves the \textbf{best} accuracy because both FetchSGD and Top-k require more overhead for compression and PVF can ensure the intact aggregation of all entries in the original vector.

\begin{table*}[t!]
\caption{Comparison of secure aggregation protocols computation time (in milliseconds) and communication cost (in KB) before and after integrating \sysname\ for a single round. The number of users $n$ is set to 100, with $\lambda$ set as 100. \textbf{w/} denotes the corresponding protocol integrated with \sysname.}
\centering
\resizebox{\textwidth}{!}{\begin{tabular}{ccccccccccc}
\toprule
\multicolumn{1}{c|}{Vector Length $(m)$}& \multicolumn{5}{c|} {100k} & \multicolumn{5}{c} {500k}\\
\cmidrule[0.5pt]{1-11}
\multicolumn{1}{c|}{Operation}&Each user & Server& Server & \multicolumn{1}{c|}{Server} &Comm. Cost& \multicolumn{1}{|c}{Each user}& Server& Server& \multicolumn{1}{c|}{Server}&Comm. Cost\\
\multicolumn{1}{c|}{Dropout rate $(\eta)$} & 0\%& 0\% & 10\% & \multicolumn{1}{c|}{30\%} & 10\%& \multicolumn{1}{|c} {0\% } & 0\% & 10\%& \multicolumn{1}{c|}{30\%}& 10\%\\
\midrule[0.5pt]
\multicolumn{1}{c|}{PPDL} & 139010 & 29759 & 27418& \multicolumn{1}{c|}{21789}&\multicolumn{1}{c|}{27761}&708579 & 149293& 136155& \multicolumn{1}{c|}{107998} &138679\\
\multicolumn{1}{c|}{w/ \sysname}& 1401 & 304 & 283& \multicolumn{1}{c|}{225}&\multicolumn{1}{c|}{859}& 7128 & 1571& 1446& \multicolumn{1}{c|}{1133}&4291\\
\multicolumn{1}{c|}{Improvement Factor}& $\color{deepgreen}\bm{99.2\times}$ & $\color{deepgreen}\bm{97.9\times}$ & $\color{deepgreen}\bm{96.8\times}$ & \multicolumn{1}{c|}{$\color{deepgreen}\bm{96.8\times}$} &\multicolumn{1}{c|}{$\color{deepgreen}\bm{32.3\times}$}&$\color{deepgreen}\bm{99.4\times}$&$\color{deepgreen}\bm{95.0\times}$&$\color{deepgreen}\bm{94.1\times}$ &\multicolumn{1}{c|}{$\color{deepgreen}\bm{95.3\times}$} &$\color{deepgreen}\bm{32.3\times}$ \\
\midrule[0.5pt]
\multicolumn{1}{c|}{EPPFL}& 3464 & 755163& 754261 & \multicolumn{1}{c|}{749279} &\multicolumn{1}{c|}{9962}& 17582& 3763769 & 3732358 & \multicolumn{1}{c|}{3745663} &49798\\
\multicolumn{1}{c|}{w/ \sysname}& 44 & 7686& 7642 & \multicolumn{1}{c|}{7711} &\multicolumn{1}{c|}{681}& 219& 38833 & 38924 & \multicolumn{1}{c|}{39612}&3400 \\
\multicolumn{1}{c|}{Improvement Factor}&$\color{deepgreen}\bm{78.7\times}$ &$\color{deepgreen}\bm{98.2\times}$&$\color{deepgreen}\bm{98.7\times}$&\multicolumn{1}{c|}{$\color{deepgreen}\bm{97.2\times}$}&\multicolumn{1}{c|}{$\color{deepgreen}\bm{14.6\times}$}&$\color{deepgreen}\bm{80.3\times}$&$\color{deepgreen}\bm{97.0\times}$&$\color{deepgreen}\bm{95.9\times}$&\multicolumn{1}{c|}{$\color{deepgreen}\bm{94.6\times}$}&$\color{deepgreen}\bm{14.6\times}$\\
\midrule[0.5pt]
\multicolumn{1}{c|}{NIVP-DS}& 15 & 386 & 372& \multicolumn{1}{c|}{331}&\multicolumn{1}{c|}{586}& 72 & 1902& 1834& \multicolumn{1}{c|}{1789}&2932\\
\multicolumn{1}{c|}{w/ \sysname}& 13 & 10& 9& \multicolumn{1}{c|}{12} &\multicolumn{1}{c|}{586}& 63 & 49& 44& \multicolumn{1}{c|}{39}&2931\\
\multicolumn{1}{c|}{Improvement Factor}&$\color{deepgreen}\bm{1.2\times}$&$\color{deepgreen}\bm{38.6\times}$&$\color{deepgreen}\bm{41.3\times}$ &\multicolumn{1}{c|}{$\color{deepgreen}\bm{27.6\times}$}&\multicolumn{1}{c|}{$\setminus$}&$\color{deepgreen}\bm{1.1\times}$&$\color{deepgreen}\bm{38.9\times}$&$\color{deepgreen}\bm{41.7\times}$&\multicolumn{1}{c|}{$\color{deepgreen}\bm{45.9\times}$}&$\setminus$\\
\midrule[0.5pt]
\multicolumn{1}{c|}{PracAgg} & 13702& 14249 & 139936 & \multicolumn{1}{c|}{307031}&\multicolumn{1}{c|}{785} & 73335& 71607 & 697950& \multicolumn{1}{c|}{1515347} &3910\\
\multicolumn{1}{c|}{w/ \sysname}& 177& 187 & 1472 & \multicolumn{1}{c|}{3207} &\multicolumn{1}{c|}{785}& 779& 794 & 7063& \multicolumn{1}{c|}{15686} &3910\\
\multicolumn{1}{c|}{Improvement Factor}&$\color{deepgreen}\bm{77.4\times}$&$\color{deepgreen}\bm{76.2\times}$ &$\color{deepgreen}\bm{95.1\times}$&\multicolumn{1}{c|}{$\color{deepgreen}\bm{95.7\times}$}&\multicolumn{1}{c|}{$\setminus$}&$\color{deepgreen}\bm{94.1\times}$&$\color{deepgreen}\bm{90.2\times}$&$\color{deepgreen}\bm{98.8\times}$&\multicolumn{1}{c|}{$\color{deepgreen}\bm{96.6\times}$}&$\setminus$\\
\midrule[0.5pt]
\multicolumn{1}{c|}{PracAgg+} & 2773 & 13973 & 39735& \multicolumn{1}{c|}{69623}& \multicolumn{1}{c|}{782}&14487& 70359 & 197540& \multicolumn{1}{c|}{347055}&3908\\
\multicolumn{1}{c|}{w/ \sysname}& 38 & 159 & 410& \multicolumn{1}{c|}{724}&\multicolumn{1}{c|}{782}& 179& 787 & 2026& \multicolumn{1}{c|}{4319}&3908\\
\multicolumn{1}{c|}{Improvement Factor}&$\color{deepgreen}\bm{72.9\times}$&$\color{deepgreen}\bm{87.9\times}$&$\color{deepgreen}\bm{97.0\times}$&\multicolumn{1}{c|}{$\color{deepgreen}\bm{96.2\times}$} &\multicolumn{1}{c|}{$\setminus$}&$\color{deepgreen}\bm{80.9\times}$&$\color{deepgreen}\bm{89.4\times}$&$\color{deepgreen}\bm{97.5\times}$&\multicolumn{1}{c|}{$\color{deepgreen}\bm{80.4\times}$}&$\setminus$\\
\midrule[0.5pt]
\multicolumn{1}{c|}{EffiAgg} & 1227 & 537771& 537329 & \multicolumn{1}{c|}{539814} &\multicolumn{1}{c|}{783}& 6014 & 2662161 & 2730404 & \multicolumn{1}{c|}{2637768} &3908\\
\multicolumn{1}{c|}{w/ \sysname}& 32 & 5677& 5702 & \multicolumn{1}{c|}{6135} &\multicolumn{1}{c|}{783}& 99 & 27175 & 27440 & \multicolumn{1}{c|}{28068} &3908\\
\multicolumn{1}{c|}{Improvement Factor}&$\color{deepgreen}\bm{38.3\times}$ &$\color{deepgreen}\bm{94.7\times}$&$\color{deepgreen}\bm{94.2\times}$&\multicolumn{1}{c|}{$\color{deepgreen}\bm{88.0\times}$}&\multicolumn{1}{c|}{$\setminus$}&$\color{deepgreen}\bm{60.7\times}$&$\color{deepgreen}\bm{98.0\times}$&$\color{deepgreen}\bm{99.5\times}$&\multicolumn{1}{c|}{$\color{deepgreen}\bm{94.0\times}$}&$\setminus$\\
\midrule[0.5pt]
\multicolumn{1}{c|}{LPPFedL} & 176& 63& 59 & \multicolumn{1}{c|}{46}&\multicolumn{1}{c|}{1564} & 846& 319 & 291 & \multicolumn{1}{c|}{218} &7814\\
\multicolumn{1}{c|}{w/ \sysname}& 22 & 21& 20 & \multicolumn{1}{c|}{17} &\multicolumn{1}{c|}{1564}& 53 & 77& 70& \multicolumn{1}{c|}{61}&7814\\
\multicolumn{1}{c|}{Improvement Factor}&$\color{deepgreen}\bm{8.0\times}$ &$\color{deepgreen}\bm{3.0\times}$&$\color{deepgreen}\bm{3.0\times}$&\multicolumn{1}{c|}{$\color{deepgreen}\bm{2.7\times}$}&\multicolumn{1}{c|}{$\setminus$}&$\color{deepgreen}\bm{16.0\times}$&$\color{deepgreen}\bm{4.1\times}$&$\color{deepgreen}\bm{4.2\times}$&\multicolumn{1}{c|}{$\color{deepgreen}\bm{3.6\times}$}&$\setminus$\\
\bottomrule
\end{tabular}}
\label{tab:compari}
\end{table*}
\begin{figure}[!htbp]
    \centering
    \includegraphics[scale=0.3]{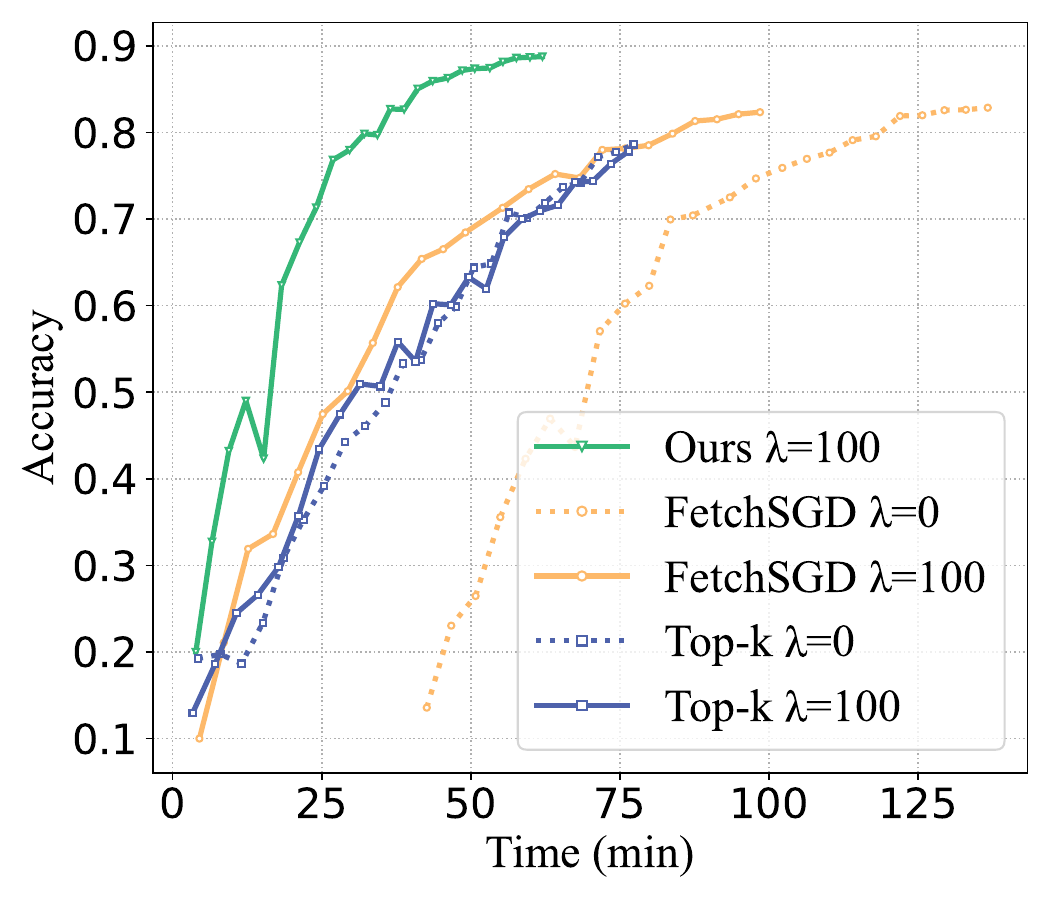}
    \caption{Comparison with compression-based techniques, with $n=100,\eta=5\%$.}
    \label{fig:compression}
\end{figure}

\begin{figure}[!h]
    \centering
    \includegraphics[scale=0.33]{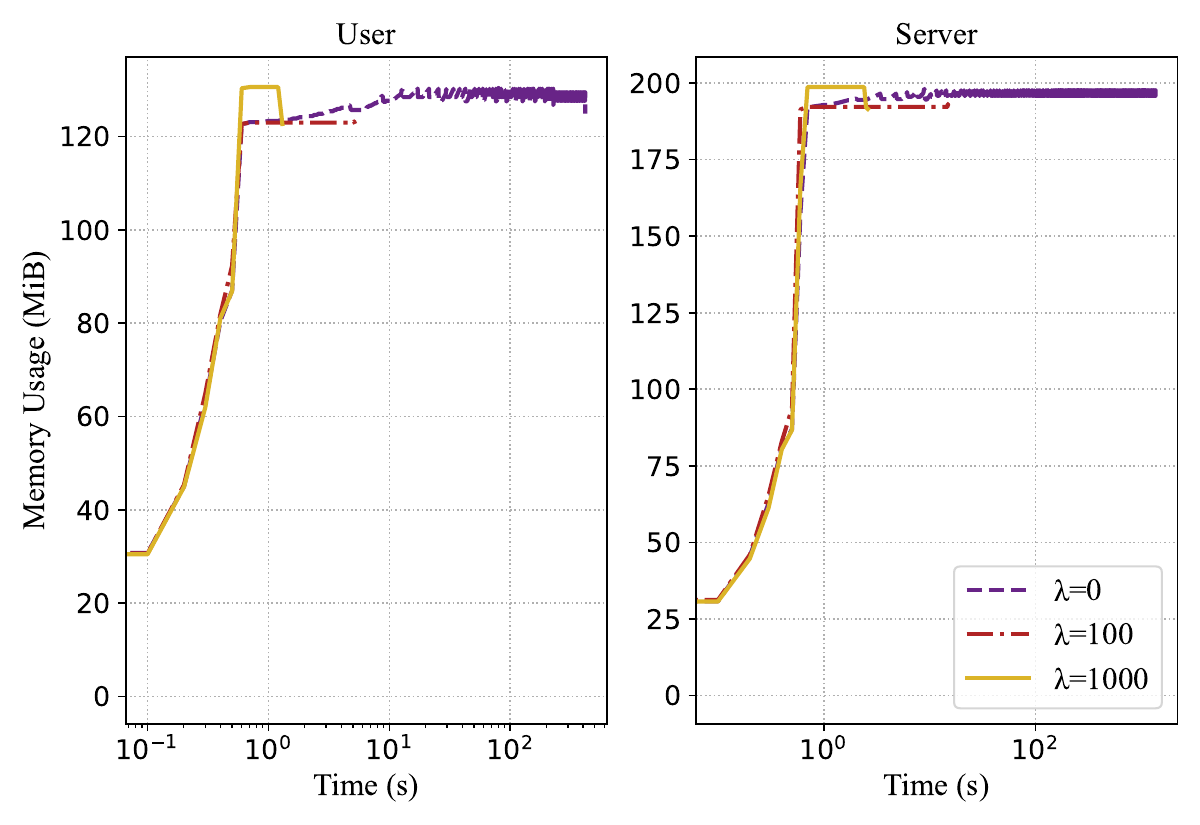}
    \caption{Memory usage with different $\lambda$. $n=100$, $m=100k$, and $\eta=10\%$.}
    \label{fig:memory}
\end{figure}

\textbf{Memory Usage}.
\label{app:memory}
In \cref{fig:memory}, we evaluate the memory usage of each user and $\mathcal{S}$ during PracAgg without PVF ($\lambda=0$) and with PVF ($\lambda=100,1000$). \textbf{For each user}, the increased memory usage for $\lambda=1000$ compared to $\lambda=100$ is due to the need to store a larger transformation matrix ($\lambda \times \lambda$) required by PVF. The increase in memory usage for $\lambda=0$ compared to $\lambda=100$ is because PVF reduces the number of random numbers generated ($nm \rightarrow \frac{nm}{\lambda}$) and decreases the scale of vector addition computations ($m \rightarrow \frac{m}{\lambda}$) during the masking process. \textbf{For $\mathcal{S}$}, all three methods require summing up the vectors uploaded by all users, with a memory usage of approximately $O(mn)$, so the additional $O(\lambda^2)$ overhead introduced by PVF is negligible. The increased memory usage for $\lambda=0$ compared to $\lambda=100$ is due to PVF reducing the number of random numbers generated ($(\eta(1-\eta)n^2+(1-\eta)n)m \rightarrow \frac{(\eta(1-\eta)n^2+(1-\eta)n)m}{\lambda}$) and decreasing the scale of vector addition computations during the unmasking process.

\textbf{For FL of LLM}.
\label{app:fedllm}
LLMs have a profound impact on the entire AI research community due to their excellent contextual learning and instruction following ability~\cite{zhao2023survey}. Given privacy concerns, training (or fine-tuning) LLM in a federated setting has been explored~\cite{hilmkil2021scaling,ye2024openfedllm}. In the context of LLMs, during the aggregation process, the length of user original vectors reaches billions. Taking Llama2-7B~\cite{touvron2023llama} as an example, assuming 1\% of the parameters need to be updated during the fine-tuning, which is 700M, the computational cost of using a general secure aggregation scheme is unimaginable, making \sysname\ particularly important. As shown in Table~\ref{tab:llm}, the time required to train one round without using PVF can basically meet the requirement of training 100 rounds with PVF. At present, FL for LLM among a multitude of lightweight clients is unrealistic. The SOTA scheme OpenFedLLM~\cite{ye2024openfedllm} involves only $n\in\{2,4,5\}$ clients per round (using Llama2-7B). And PVF can contribute to future FedLLMs and involving more clients.

\begin{figure*}[!t]
    \centering
    \includegraphics[width=0.65\textwidth]{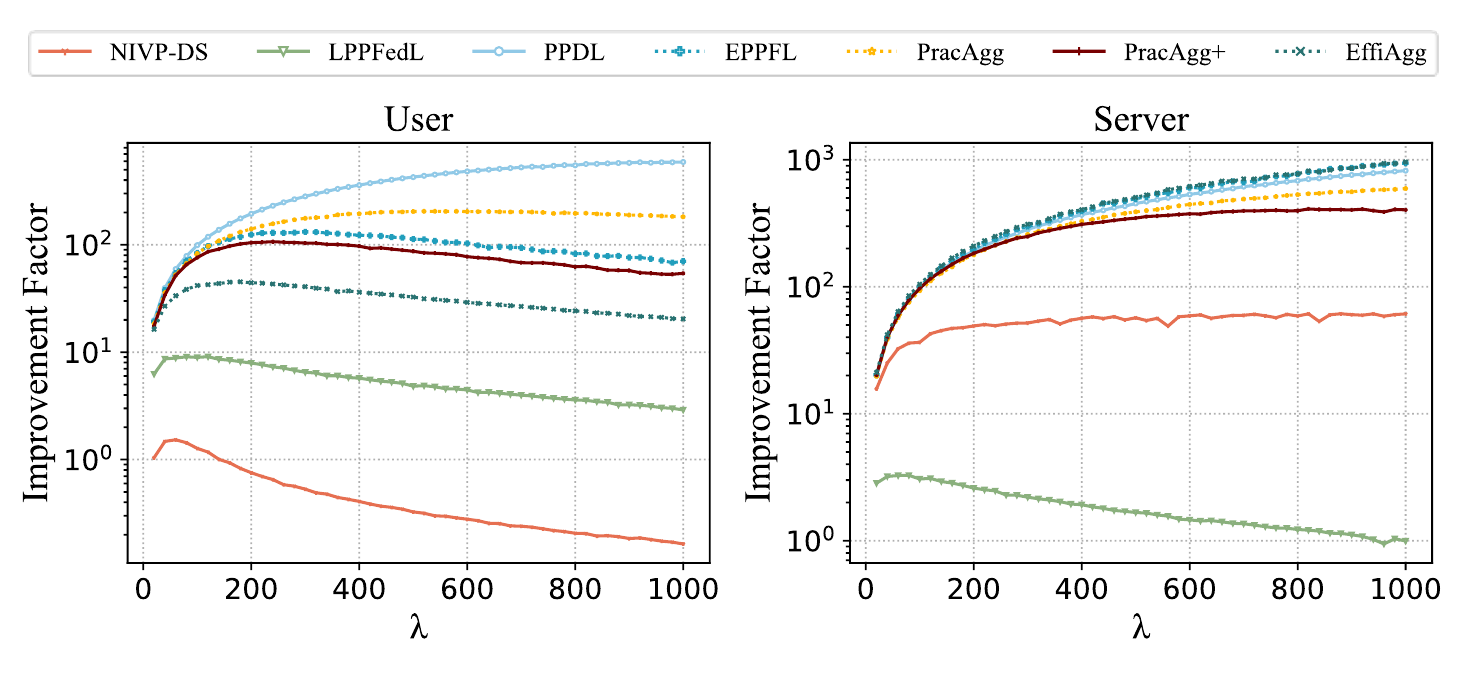}
    \caption{Computation acceleration for secure aggregation brought by different $\lambda$ in \sysname. $n=100$, $m=100k$, and $\eta=10\%$.}
    \label{fig:lambda}
\end{figure*}

\begin{table}[h!]
\caption{
Estimated computational overhead per round with and without PVF for fine-tuning Llama2-7B. $n=100$, $m=700M$, $\lambda=100$, and $\eta=0\%$.}
\centering
\resizebox{0.7\linewidth}{!}{
 \begin{tabular}{ccccc}
\toprule
 \multirow{2}{*}{Scheme}&\multicolumn{2}{|c|}{Each user}& \multicolumn{2}{c}{Server}\\
\cmidrule[0.5pt](lr){2-3} \cmidrule[0.5pt](lr){4-5} 
{}&\multicolumn{1}{|c}{w/o}& \multicolumn{1}{c|}{w/} & w/o & \multicolumn{1}{c}{w/}\\
\midrule[0.5pt]
\multicolumn{1}{c|}{PPDL} &$\sim$300h&\multicolumn{1}{r|}{$\sim$3h}&$\sim$60h&\multicolumn{1}{c}{$\sim$0.6h}\\
\multicolumn{1}{c|}{PracAgg} &$\sim$27h&\multicolumn{1}{r|}{$\sim$0.27h}&$\sim$27h&\multicolumn{1}{c}{$\sim$0.27h}\\
\multicolumn{1}{c|}{PracAgg+}&$\sim$5h&\multicolumn{1}{r|}{$\sim$0.05h}&$\sim$27h&\multicolumn{1}{c}{$\sim$0.27h}\\
\bottomrule
\end{tabular}}
\label{tab:llm}
\end{table}

\subsection{Ablation Study}
\label{ablation}

\textbf{$\bm{\lambda}$. }
We conduct experiments to analyze the variation in the acceleration effect of \sysname\ on all baselines under different $\lambda$ values, as depicted in Figure \ref{fig:lambda}. \sysname\ exhibits a more pronounced acceleration effect for participants with substantial original computational overhead, such as PPDL users. For certain SAPs like EffiAgg, the improvement factor of $\mathcal{S}$ can reach up to $1000\times$. As $\lambda$ increases, the improvement factor of \sysname\ gradually stabilizes and may experience a slight decline in certain SAPs. This trend emerges because when $\lambda$ becomes sufficiently large, the primary computational overhead in secure aggregation shifts from \textbf{masking-related} overhead to \textbf{interaction-related} overhead like secret sharing. It's worth noting that when $\lambda$ reaches a certain threshold, for certain entities with relatively low primary computational load, such as users in NIVP-DS and servers in LPPFedL, integrating \sysname\ may increase their computational overhead (improvement factor less than $1$). This occurs because the computation cost incurred by the linear transformation in \sysname\ becomes notable. But the gains from \sysname\ for the servers in NIVP-DS and users in LPPFedL remain enticing, rendering the cost of the \sysname\ linear transformation negligible.

\textbf{Disrupting Variables Extension. } 
We focus on the impact of DVE on the model's performance in this part. In the 32-bit input space, $\epsilon$ in \cref{eq:sigma} is set to 8783. And we explore the impact of DVE on model performance in two practical tasks: \textbf{(i) Image classification}, with LeNet model and MNIST dataset; \textbf{(ii) Movie recommendation}, using FedMF~\cite{chai2020secure} model with the item embedding size of 32, i.e. 118,592 parameters in total, and ML-1M~\cite{harper2015movielens} dataset. As shown in Figure~\ref{fig:dp}, the impact of using PracAgg integrated with DVE on model performance is negligible.
\begin{figure}[!h]
    \centering
    \includegraphics[scale=0.4]{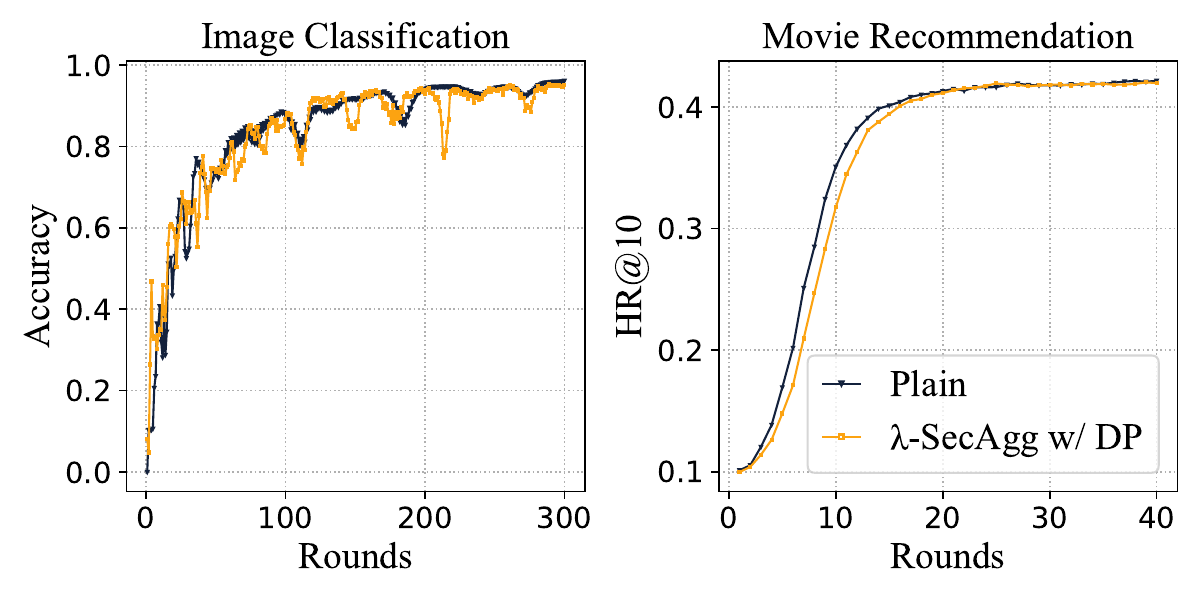}
    \caption{The impact of DVE on model performance in practical applications. In these two tasks, $n=100$, $\eta=100$, and $\lambda$=100. }
    \label{fig:dp}
\end{figure}

{$\bm{n}$.} Figure \ref{fig:num} displays the speedup effect of \sysname\ across different $n$. The speedup remains consistently high for PracAgg+ and PPDL across various $n$. However, as $n$ rises, the speedup for PracAgg diminishes. This occurs because PracAgg necessitates multiple secret sharings for each user across the entire user set, along with multiple secret reconstructions by $\mathcal{S}$. These interaction-related overhead increases more significantly as $n$ increases. Nonetheless, even when $n$ reaches $1000$, the gain brought by \sysname\ for PracAgg servers remains above $\bm{85\times}$, with user gain still exceeding $\bm{55\times}$.

\begin{figure}[!h]
    \centering
    \includegraphics[scale=0.4]{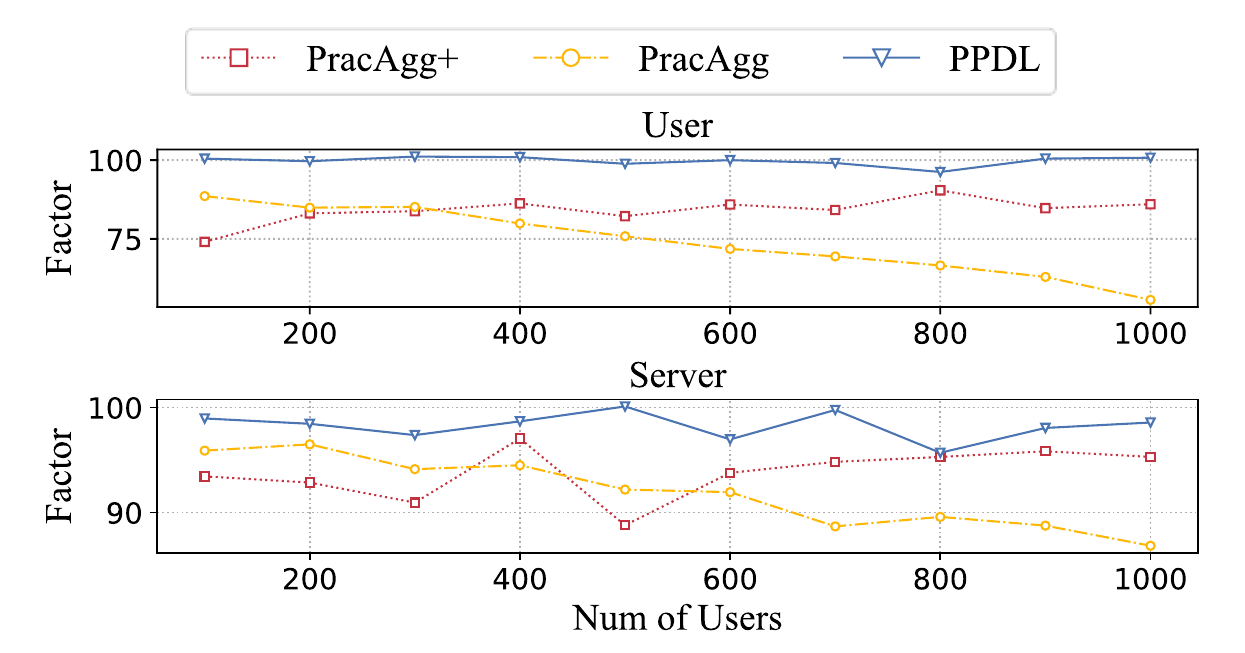}
    \caption{The gain introduced by \sysname\ with respect to the number of users, with $\eta=10\%$. $m=100k$, and $\lambda=100$.}
    \label{fig:num}
\end{figure}

\textbf{Transformation. }The additional overhead of integrating \sysname\ primarily arises from \cref{eq:y,eq:k,eq:lineareqs}. 
Figure~\ref{fig:trans} illustrates the \textbf{overall} additional computation costs of PVF concerning different $m$ and $\lambda$. It is evident that when $\lambda=100$ and the original vector length is substantial (reaching $200k$), the computation time is less than 15\,ms, nearly negligible compared to the computational time saved by \sysname. 
As $\lambda$ gradually increases while vector length remains fixed, the number of loops during freezing and thawing computations decreases in the beginning because of the reduced number of groups. When $\lambda$ exceeds 60 and continues to increase, the scale of matrix-vector multiplication gradually rises, consequently increasing the computation cost. However, this cost remains consistently within the range of tens of milliseconds, posing a negligible computation burden.
\begin{figure}[!h]
    \centering
    \includegraphics[scale=0.4]{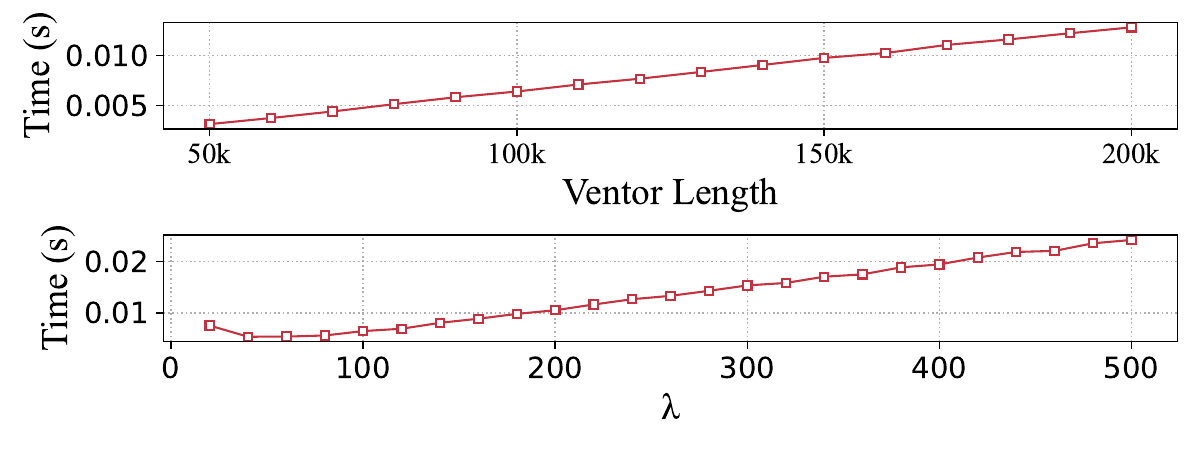}
    \caption{Additional computational overhead incurred by \sysname\ itself under different $\lambda$ and vector lengths.}
    \label{fig:trans}
\end{figure}

\textbf{The impact of padding}.
\label{app:padding}
\sysname\ necessitates padding the original vector, thereby increasing the vector size. 
In theory, the additional cost that padding introduces in computation and communication does not exceed $\frac{\lambda}{m}$ of the original, as \textbf{the maximum value of padding length} is $\lambda$. Table \ref{tab:padding} showcases the impact of no padding versus padding $\lambda$ entries on the computation and communication overhead, where $\frac{\lambda}{m}=0.001$. It's evident that the overhead induced by padding is almost negligible. 
To strictly adhere to the principle of ``\textbf{not increasing any communication overhead}'', we present a method to avoid padding. We extract the first $\left \lfloor \frac{m}{\lambda} \right \rfloor \lambda$ entries from the original vector and apply PVF to them. The remaining entries, which are fewer than $\lambda$, are appended to the key vector $\bm{k}$ for participation in the secure aggregation. This approach allows us to obtain the aggregate result of the entire original vector while eliminating the need for padding.
\begin{table}[h!]
\caption{
Comparison of overhead with and without padding. $n=100$, $m=100k$, $\lambda=100$, and $\eta=10\%$.}
\centering

 \resizebox{\linewidth}{!}{\begin{tabular}{ccccccc}
\toprule
 \multirow{2}{*}{Scheme}&\multicolumn{2}{|c|}{User comp. (ms)}& \multicolumn{2}{c|}{Server comp. (ms)}&\multicolumn{2}{c}{Comm. cost (KB)}\\
\cmidrule[0.5pt](lr){2-3} \cmidrule[0.5pt](lr){4-5} \cmidrule[0.5pt](lr){6-7}
{}&\multicolumn{1}{|c}{No pad}& \multicolumn{1}{c|}{Pad} & No pad & \multicolumn{1}{c|}{Pad}& No pad & Pad\\
\midrule[0.5pt]
\multicolumn{1}{c|}{PPDL} &1402&\multicolumn{1}{r|}{1403, $\uparrow 1$}&303&\multicolumn{1}{r|}{304, $\uparrow 1$}&859 &860, $\uparrow 1$\\
\multicolumn{1}{c|}{EPPFL}&41&\multicolumn{1}{r|}{42, $\uparrow 1$}&7337&\multicolumn{1}{r|}{7351, $\uparrow 14$}&681&681, $\uparrow 0$\\
\multicolumn{1}{c|}{NIVP-DS} &13&\multicolumn{1}{r|}{13, $\uparrow 0$}&10&\multicolumn{1}{r|}{10, $\uparrow 0$}&587&587, $\uparrow 0$\\
\multicolumn{1}{c|}{PracAgg} &162&\multicolumn{1}{r|}{164, $\uparrow 2$}&1456&\multicolumn{1}{r|}{1463, $\uparrow 7$}&785&785, $\uparrow 0$\\
\multicolumn{1}{c|}{PracAgg+}&38&\multicolumn{1}{r|}{39, $\uparrow 1$}&405&\multicolumn{1}{r|}{409, $\uparrow 4$}&783&783, $\uparrow 0$\\
\multicolumn{1}{c|}{EffiAgg}&29&\multicolumn{1}{r|}{32, $\uparrow 3$}&5404&\multicolumn{1}{r|}{5482, $\uparrow 78$}&783&783, $\uparrow 0$\\
\multicolumn{1}{c|}{LPPFedL}&19&\multicolumn{1}{r|}{19, $\uparrow 0$}&20&\multicolumn{1}{r|}{20, $\uparrow 0$}&1564&1564, $\uparrow 0$\\
\bottomrule
\end{tabular}}
\label{tab:padding}
\end{table}

\section{Conclusion}
\label{sec:conclusion}
We present a new perspective aimed at mitigating the formidable computation overhead of SAPs by reducing the number of involved entries while ensuring intact secure aggregation of original vectors. Based on this, we propose \sysname, a concrete portable solution. After integrating with \sysname, $\lambda$-SecAgg involves only $\frac{1}{\lambda}$ of original vectors. Moreover, we introduce the disrupting variables extension to improve security. Extensive experiments showcase the remarkable improvements in acceleration and communication brought by \sysname\ and its portability. Consequently, our method undoubtedly renders SAPs genuinely feasible, promising inspiration for future research endeavors.




\newpage
\clearpage
\section*{Ethics Considerations}
Our proposal has no potential negative impact on society and we believe the research is done ethically. Instead, it effectively safeguards the privacy of users while enhancing the efficiency of secure aggregation.
\newpage
\clearpage
\normalem
\bibliographystyle{plain}
\bibliography{ref}

\begin{thebibliography}{10}

\bibitem{aono2017privacy}
Yoshinori Aono, Takuya Hayashi, Lihua Wang, Shiho Moriai, et~al.
\newblock Privacy-preserving deep learning via additively homomorphic encryption.
\newblock {\em IEEE Transactions on Information Forensics and Security}, 13(5):1333--1345, 2017.

\bibitem{Bell2020SecureSA}
James Bell, Kallista~A. Bonawitz, Adri{\`a} Gasc{\'o}n, Tancr{\`e}de Lepoint, and Mariana Raykova.
\newblock Secure single-server aggregation with (poly)logarithmic overhead.
\newblock {\em Proceedings of the 2020 ACM SIGSAC Conference on Computer and Communications Security}, 2020.

\bibitem{bell2023acorn}
James Bell, Adri{\`a} Gasc{\'o}n, Tancr{\`e}de Lepoint, Baiyu Li, Sarah Meiklejohn, Mariana Raykova, and Cathie Yun.
\newblock $\{$ACORN$\}$: Input validation for secure aggregation.
\newblock In {\em 32nd USENIX Security Symposium (USENIX Security 23)}, pages 4805--4822, 2023.

\bibitem{bellare2000authenticated}
Mihir Bellare and Chanathip Namprempre.
\newblock Authenticated encryption: Relations among notions and analysis of the generic composition paradigm.
\newblock In {\em Advances in Cryptology—ASIACRYPT 2000: 6th International Conference on the Theory and Application of Cryptology and Information Security Kyoto, Japan, December 3--7, 2000 Proceedings 6}, pages 531--545. Springer, 2000.

\bibitem{boer2020secure}
Derian Boer and Stefan Kramer.
\newblock Secure sum outperforms homomorphic encryption in (current) collaborative deep learning.
\newblock {\em arXiv preprint arXiv:2006.02894}, 2020.

\bibitem{Bonawitz2017PracticalSA}
Keith Bonawitz, Vladimir Ivanov, Ben Kreuter, Antonio Marcedone, H.~B. McMahan, Sarvar Patel, Daniel Ramage, Aaron Segal, and Karn Seth.
\newblock Practical secure aggregation for privacy-preserving machine learning.
\newblock {\em Proceedings of the 2017 ACM SIGSAC Conference on Computer and Communications Security}, 2017.

\bibitem{chai2020secure}
Di~Chai, Leye Wang, Kai Chen, and Qiang Yang.
\newblock Secure federated matrix factorization.
\newblock {\em IEEE Intelligent Systems}, 36(5):11--20, 2020.

\bibitem{deng2012mnist}
Li~Deng.
\newblock The mnist database of handwritten digit images for machine learning research.
\newblock {\em IEEE Signal Processing Magazine}, 29(6):141--142, 2012.

\bibitem{elkordy2023much}
Ahmed~Roushdy Elkordy, Jiang Zhang, Yahya~H Ezzeldin, Konstantinos Psounis, and Salman Avestimehr.
\newblock How much privacy does federated learning with secure aggregation guarantee?
\newblock {\em Proceedings on Privacy Enhancing Technologies}, 2023.

\bibitem{Eltaras2023EfficientVP}
Tamer~Ahmed Eltaras, Farida Sabry, Wadha Labda, Khawla Alzoubi, and Qutaibah Ahmedeltaras.
\newblock Efficient verifiable protocol for privacy-preserving aggregation in federated learning.
\newblock {\em IEEE Transactions on Information Forensics and Security}, 18:2977--2990, 2023.

\bibitem{Ergun2021SparsifiedSA}
Irem Ergun, Hasin~Us Sami, and Basak Guler.
\newblock Sparsified secure aggregation for privacy-preserving federated learning.
\newblock {\em ArXiv}, abs/2112.12872, 2021.

\bibitem{geiping2020inverting}
Jonas Geiping, Hartmut Bauermeister, Hannah Dr{\"o}ge, and Michael Moeller.
\newblock Inverting gradients-how easy is it to break privacy in federated learning?
\newblock {\em Advances in Neural Information Processing Systems}, 33:16937--16947, 2020.

\bibitem{Geyer2017DifferentiallyPF}
Robin~C. Geyer, Tassilo Klein, and Moin Nabi.
\newblock Differentially private federated learning: A client level perspective.
\newblock {\em ArXiv}, abs/1712.07557, 2017.

\bibitem{Guo2021VeriFLCA}
Xiaojie Guo, Zheli Liu, Jin Li, Jiqiang Gao, Boyu Hou, Changyu Dong, and Thar Baker.
\newblock Verifl: Communication-efficient and fast verifiable aggregation for federated learning.
\newblock {\em IEEE Transactions on Information Forensics and Security}, 16:1736--1751, 2021.

\bibitem{Hahn2023VerSAVS}
Changhee Hahn, Hodong Kim, Minjae Kim, and Junbeom Hur.
\newblock Versa: Verifiable secure aggregation for cross-device federated learning.
\newblock {\em IEEE Transactions on Dependable and Secure Computing}, 20:36--52, 2023.

\bibitem{harper2015movielens}
F~Maxwell Harper and Joseph~A Konstan.
\newblock The movielens datasets: History and context.
\newblock {\em Acm transactions on interactive intelligent systems (tiis)}, 5(4):1--19, 2015.

\bibitem{he2017neural}
Xiangnan He, Lizi Liao, Hanwang Zhang, Liqiang Nie, Xia Hu, and Tat-Seng Chua.
\newblock Neural collaborative filtering.
\newblock In {\em Proceedings of the 26th international conference on world wide web}, pages 173--182, 2017.

\bibitem{hilmkil2021scaling}
Agrin Hilmkil, Sebastian Callh, Matteo Barbieri, Leon~Ren{\'e} S{\"u}tfeld, Edvin~Listo Zec, and Olof Mogren.
\newblock Scaling federated learning for fine-tuning of large language models.
\newblock In {\em International Conference on Applications of Natural Language to Information Systems}, pages 15--23. Springer, 2021.

\bibitem{kadhe2020fastsecagg}
Swanand Kadhe, Nived Rajaraman, O~Ozan Koyluoglu, and Kannan Ramchandran.
\newblock Fastsecagg: Scalable secure aggregation for privacy-preserving federated learning.
\newblock {\em arXiv preprint arXiv:2009.11248}, 2020.

\bibitem{kalikinkar2018nikebased}
Mandal Kalikinkar, Gong Guang, and Liu Chuyi.
\newblock Nike-based fast privacy-preserving high-dimensional data aggregation for mobile devices.
\newblock {\em University of Waterloo, Waterloo, ON, Canada, Tech. Rep. CACR}, 10:2018, 2018.

\bibitem{lecun1998gradient}
Yann LeCun, L{\'e}on Bottou, Yoshua Bengio, and Patrick Haffner.
\newblock Gradient-based learning applied to document recognition.
\newblock {\em Proceedings of the IEEE}, 86(11):2278--2324, 1998.

\bibitem{Li2022EfficientPF}
Yiran Li, Hongwei Li, Guowen Xu, Xiaoming Huang, and Rongxing Lu.
\newblock Efficient privacy-preserving federated learning with unreliable users.
\newblock {\em IEEE Internet of Things Journal}, 9:11590--11603, 2022.

\bibitem{liu2021machine}
Bo~Liu, Ming Ding, Sina Shaham, Wenny Rahayu, Farhad Farokhi, and Zihuai Lin.
\newblock When machine learning meets privacy: A survey and outlook.
\newblock {\em ACM Computing Surveys (CSUR)}, 54(2):1--36, 2021.

\bibitem{Liu2022EfficientDA}
Ziyao Liu, Jiale Guo, Kwok-Yan Lam, and Jun Zhao.
\newblock Efficient dropout-resilient aggregation for privacy-preserving machine learning.
\newblock {\em IEEE Transactions on Information Forensics and Security}, 18:1839--1854, 2022.

\bibitem{liu2022privacy}
Ziyao Liu, Jiale Guo, Wenzhuo Yang, Jiani Fan, Kwok-Yan Lam, and Jun Zhao.
\newblock Privacy-preserving aggregation in federated learning: A survey.
\newblock {\em IEEE Transactions on Big Data}, 2022.

\bibitem{Liu2023LongTermPA}
Ziyao Liu, Hsiao-Ying Lin, and Yamin Liu.
\newblock Long-term privacy-preserving aggregation with user-dynamics for federated learning.
\newblock {\em IEEE Transactions on Information Forensics and Security}, 18:2398--2412, 2023.

\bibitem{liu2022sash}
Zizhen Liu, Si~Chen, Jing Ye, Junfeng Fan, Huawei Li, and Xiaowei Li.
\newblock Sash: Efficient secure aggregation based on shprg for federated learning.
\newblock In {\em Uncertainty in Artificial Intelligence}, pages 1243--1252. PMLR, 2022.

\bibitem{lu2023top}
Shiwei Lu, Ruihu Li, Wenbin Liu, Chaofeng Guan, and Xiaopeng Yang.
\newblock Top-k sparsification with secure aggregation for privacy-preserving federated learning.
\newblock {\em Computers \& Security}, 124:102993, 2023.

\bibitem{lycklama2023rofl}
Hidde Lycklama, Lukas Burkhalter, Alexander Viand, Nicolas K{\"u}chler, and Anwar Hithnawi.
\newblock Rofl: Robustness of secure federated learning.
\newblock In {\em 2023 IEEE Symposium on Security and Privacy (SP)}, pages 453--476. IEEE, 2023.

\bibitem{ma2022privacy}
Jing Ma, Si-Ahmed Naas, Stephan Sigg, and Xixiang Lyu.
\newblock Privacy-preserving federated learning based on multi-key homomorphic encryption.
\newblock {\em International Journal of Intelligent Systems}, 37(9):5880--5901, 2022.

\bibitem{Ma2023FlamingoMS}
Yiping Ma, Jess Woods, Sebastian Angel, Antigoni Polychroniadou, and Tal Rabin.
\newblock Flamingo: Multi-round single-server secure aggregation with applications to private federated learning.
\newblock {\em 2023 IEEE Symposium on Security and Privacy (SP)}, pages 477--496, 2023.

\bibitem{mcmahan2017communication}
Brendan McMahan, Eider Moore, Daniel Ramage, Seth Hampson, and Blaise~Aguera y~Arcas.
\newblock Communication-efficient learning of deep networks from decentralized data.
\newblock In {\em Artificial intelligence and statistics}, pages 1273--1282. PMLR, 2017.

\bibitem{Meurer_SymPy_symbolic_computing_2017}
Aaron Meurer, Christopher~P. Smith, Mateusz Paprocki, Ondřej Čertík, Sergey~B. Kirpichev, Matthew Rocklin, Amit Kumar, Sergiu Ivanov, Jason~K. Moore, Sartaj Singh, Thilina Rathnayake, Sean Vig, Brian~E. Granger, Richard~P. Muller, Fransesco Bonazzi, Harsh Gupta, Shivam Vats, Fredrik Johansson, Fabian Pedregosa, Matthew~J. Curry, Andy~R. Terrel, Štěpán Roučka, Ashutosh Saboo, Isuru Fernando, Sumith Kulal, Robert Cimrmam, and Anthony Scopatz.
\newblock {SymPy: symbolic computing in Python}.
\newblock {\em PeerJ Computer Science}, 3, January 2017.

\bibitem{Mozaffari2023EveryVC}
Hamid Mozaffari, Virat Shejwalkar, and Amir Houmansadr.
\newblock Every vote counts: Ranking-based training of federated learning to resist poisoning attacks.
\newblock In {\em USENIX Security Symposium}, 2023.

\bibitem{pasquini2022eluding}
Dario Pasquini, Danilo Francati, and Giuseppe Ateniese.
\newblock Eluding secure aggregation in federated learning via model inconsistency.
\newblock In {\em Proceedings of the 2022 ACM SIGSAC Conference on Computer and Communications Security}, pages 2429--2443, 2022.

\bibitem{pedersen1991non}
Torben~Pryds Pedersen.
\newblock Non-interactive and information-theoretic secure verifiable secret sharing.
\newblock In {\em Annual international cryptology conference}, pages 129--140. Springer, 1991.

\bibitem{rathee2023elsa}
Mayank Rathee, Conghao Shen, Sameer Wagh, and Raluca~Ada Popa.
\newblock Elsa: Secure aggregation for federated learning with malicious actors.
\newblock In {\em 2023 IEEE Symposium on Security and Privacy (SP)}, pages 1961--1979. IEEE, 2023.

\bibitem{regev2009lattices}
Oded Regev.
\newblock On lattices, learning with errors, random linear codes, and cryptography.
\newblock {\em Journal of the ACM (JACM)}, 56(6):1--40, 2009.

\bibitem{rothchild2020fetchsgd}
Daniel Rothchild, Ashwinee Panda, Enayat Ullah, Nikita Ivkin, Ion Stoica, Vladimir Braverman, Joseph Gonzalez, and Raman Arora.
\newblock Fetchsgd: Communication-efficient federated learning with sketching.
\newblock In {\em International Conference on Machine Learning}, pages 8253--8265. PMLR, 2020.

\bibitem{shamir1979share}
Adi Shamir.
\newblock How to share a secret.
\newblock {\em Communications of the ACM}, 22(11):612--613, 1979.

\bibitem{so2023securing}
Jinhyun So, Ramy~E Ali, Ba{\c{s}}ak G{\"u}ler, Jiantao Jiao, and A~Salman Avestimehr.
\newblock Securing secure aggregation: Mitigating multi-round privacy leakage in federated learning.
\newblock In {\em Proceedings of the AAAI Conference on Artificial Intelligence}, volume~37, pages 9864--9873, 2023.

\bibitem{so2021turbo}
Jinhyun So, Ba{\c{s}}ak G{\"u}ler, and A~Salman Avestimehr.
\newblock Turbo-aggregate: Breaking the quadratic aggregation barrier in secure federated learning.
\newblock {\em IEEE Journal on Selected Areas in Information Theory}, 2(1):479--489, 2021.

\bibitem{so2022lightsecagg}
Jinhyun So, Chaoyang He, Chien-Sheng Yang, Songze Li, Qian Yu, Ramy E~Ali, Basak Guler, and Salman Avestimehr.
\newblock Lightsecagg: a lightweight and versatile design for secure aggregation in federated learning.
\newblock {\em Proceedings of Machine Learning and Systems}, 4:694--720, 2022.

\bibitem{Sotthiwat2021PartiallyEM}
Ekanut Sotthiwat, Liangli Zhen, Zengxiang Li, and Chi Zhang.
\newblock Partially encrypted multi-party computation for federated learning.
\newblock {\em 2021 IEEE/ACM 21st International Symposium on Cluster, Cloud and Internet Computing (CCGrid)}, pages 828--835, 2021.

\bibitem{stevens2022efficient}
Timothy Stevens, Christian Skalka, Christelle Vincent, John Ring, Samuel Clark, and Joseph Near.
\newblock Efficient differentially private secure aggregation for federated learning via hardness of learning with errors.
\newblock In {\em 31st USENIX Security Symposium (USENIX Security 22)}, pages 1379--1395, 2022.

\bibitem{suetin1989linear}
PK~Suetin, Alexandra~I Kostrikin, and Yu~I Manin.
\newblock {\em Linear algebra and geometry}.
\newblock CRC Press, 1989.

\bibitem{touvron2023llama}
Hugo Touvron, Louis Martin, Kevin Stone, Peter Albert, Amjad Almahairi, Yasmine Babaei, Nikolay Bashlykov, Soumya Batra, Prajjwal Bhargava, Shruti Bhosale, et~al.
\newblock Llama 2: Open foundation and fine-tuned chat models.
\newblock {\em arXiv preprint arXiv:2307.09288}, 2023.

\bibitem{wang2021protecting}
Chuanyin Wang, Cunqing Ma, Min Li, Neng Gao, Yifei Zhang, and Zhuoxiang Shen.
\newblock Protecting data privacy in federated learning combining differential privacy and weak encryption.
\newblock In {\em Science of Cyber Security: Third International Conference, SciSec 2021, Virtual Event, August 13--15, 2021, Revised Selected Papers 4}, pages 95--109. Springer, 2021.

\bibitem{Wang2023VOSAVA}
Yong Wang, Aiqing Zhang, Shu-Lin Wu, and Shui Yu.
\newblock Vosa: Verifiable and oblivious secure aggregation for privacy-preserving federated learning.
\newblock {\em IEEE Transactions on Dependable and Secure Computing}, 20:3601--3616, 2023.

\bibitem{wei2020federated}
Kang Wei, Jun Li, Ming Ding, Chuan Ma, Howard~H Yang, Farhad Farokhi, Shi Jin, Tony~QS Quek, and H~Vincent Poor.
\newblock Federated learning with differential privacy: Algorithms and performance analysis.
\newblock {\em IEEE Transactions on Information Forensics and Security}, 15:3454--3469, 2020.

\bibitem{Wei2023LightweightFL}
Zhaohui Wei, Qingqi Pei, Ning Zhang, Xuefeng Liu, Celimuge Wu, and Amirhosein Taherkordi.
\newblock Lightweight federated learning for large-scale iot devices with privacy guarantee.
\newblock {\em IEEE Internet of Things Journal}, 10:3179--3191, 2023.

\bibitem{xu2022non}
Yi~Xu, Changgen Peng, Weijie Tan, Youliang Tian, Minyao Ma, and Kun Niu.
\newblock Non-interactive verifiable privacy-preserving federated learning.
\newblock {\em Future Generation Computer Systems}, 128:365--380, 2022.

\bibitem{ye2024openfedllm}
Rui Ye, Wenhao Wang, Jingyi Chai, Dihan Li, Zexi Li, Yinda Xu, Yaxin Du, Yanfeng Wang, and Siheng Chen.
\newblock Openfedllm: Training large language models on decentralized private data via federated learning.
\newblock In {\em Proceedings of the 30th ACM SIGKDD Conference on Knowledge Discovery and Data Mining}, pages 6137--6147, 2024.

\bibitem{zhao2023survey}
Wayne~Xin Zhao, Kun Zhou, Junyi Li, Tianyi Tang, Xiaolei Wang, Yupeng Hou, Yingqian Min, Beichen Zhang, Junjie Zhang, Zican Dong, et~al.
\newblock A survey of large language models.
\newblock {\em arXiv preprint arXiv:2303.18223}, 2023.

\bibitem{zhu2019deep}
Ligeng Zhu, Zhijian Liu, and Song Han.
\newblock Deep leakage from gradients.
\newblock {\em Advances in neural information processing systems}, 32, 2019.

\end{thebibliography}
\appendix


\section{Notations}
\label{sec:notation}
The primary notations used in this work are listed in Table~\ref{tab:notations}.
\begin{table}[htbp!]
    \caption{Notations.}
    \centering
    \label{tab:notations}
    \small
    \resizebox{\linewidth}{!}{\begin{tabular}{rl}
            \toprule
            \textbf{Symbol}                           & \textbf{Definition }                                                                \\
            \midrule

            $n$                                       & Number of users                                                                     \\
            $m$                                       & Length of original vectors                                                          \\
            $\eta$                                    & Dropout rate                                                                        \\
            $\lambda$                                 & Compression factor                                                                  \\
            $\mu$                                  & The number of broken elements within every $\lambda$ elements\\
            $\mathcal{U}$                             & The user set                                                                        \\
                        $\mathcal{U'}$                             & The surviving user set                                                                        \\
            $\mathcal{S}$                             & The server                                                                          \\
            $\bm{x}^{i(t)}$                       & Original vector of $i$ in the $t$-th round (with DP-based noise)                                    \\
            $\bm{A}$                    & A $\lambda \times \lambda$ invertible matrix                                        \\
            $\bm{A}^{-1}$               & The inverse matrix of $\bm{A}$                                        \\
            $\check{\bm{A}}$            & The matrix composed only of the first $\lambda-1$ rows of $\bm{A}$      \\
            $\bm{\alpha}$                   & The vector composed of the $\lambda$-th row of $\bm{A}$                 \\
            $\check{\bm{A}}^{\mu+1}$ & $\check{\bm{A}}$ with $\mu$-security                               \\
            $\bm{\alpha}^{\mu+1}$        & $\bm{\alpha}$ with $\mu$-security                                      \\
            $SLE_{AK}(\bm{Ax})$                   & Function for computing $\bm{x}$ of $\bm{Ax}$ with the knowledge $AK$ \\
            $\bm{x}^{i(t)}_{pad}$                 & $\bm{x}^{i(t)}$ after padding                                                   \\
            $m'$                                      & Length of $\bm{x}^{i(t)}_{pad}$                                                 \\
            $\bm{d}^{i(t)}_j$                     & The $j$-th group of $\bm{x}^{i(t)}_{pad}$                                    \\
            $\bm{y}^{i(t)}$                       & Frozen vector of $\bm{x}^{i(t)}_{pad}$                                          \\
            $\bm{GN}$ & DP-based noise generated by the Gaussian mechanism \\
            $\bm{k}^{i(t)}$                       & Key vector of $\bm{x}^{i(t)}_{pad}$                                             \\
            $\bm{c}^{i(t)}$                       & Commitment vector of $\bm{x}^{i(t)}$                                            \\
            $\bm{\zeta}^{i(t)}$                       & Corresponding random vector of $\bm{c}^{i(t)}$                                  \\
            \bottomrule
        \end{tabular}}
\end{table}

\section{Optional Extensions}
\label{app:opt-extensions}
We introduce 3 optional extensions augmenting functionality and security for better portability: (i) $\mu$\textit{-security Extension}, which resists adversaries that possess prior knowledge of partial private vectors; (ii) \textit{Result Verification Extension} (RVE), which ensures the correctness of performing transformations to prevent malicious server-side computations, and (iii) \textit{User Commitment Extension} (UCE), which ensures the uniqueness of users' original vectors to prevent malicious attempts at incorrect computations. The detailed pipeline for integrating \sysname\ with all 4 extensions into SAP is depicted in Figure \ref{fig:protocol_fin}.
\begin{figure*}[t!]
    \begin{center}
        \resizebox{0.9\textwidth}{!}{\begin{tabular}{|l|}
                \hline

            \\
                \textbf{Participants}: $\mathcal{S}$ and User set $\mathcal{U}=\{u_1,u_2,\ldots,u_n\}$.                                                                                                                                                                      \\
                \textbf{Public Inputs}: $\bm{A}$, $\mu$, $\check{\bm{A}}^{\mu +1}$, $\bm{\alpha}^{\mu+1}$, $\lambda$, $\mathbb{Z}_p$, $g$ and $h$. Users' public keys for signatures $\{ sig_i^{pk}\}_{i \in \mathcal{U}}$ and the server's public key for signatures $sig_S^{pk}$.\\
                \textbf{Private Inputs}: Original vectors $\{ \bm{x}^{i(t)}\}_{i \in \mathcal{U}}$ of $t$-th iteration. Users' secret keys for signatures $\{ sig_i^{sk}\}_{i \in \mathcal{U}}$ and the server's secret key for signatures $sig_S^{sk}$.                                                                                                                                                 \\
                \textbf{Outputs}: Surviving user set $\mathcal{U}'$, $\sum_{i\in \mathcal{U}'} \bm{x}^{i(t)}$.                                                                                                                                                           \\
                $~\bullet~$\textbf{Phase 1  Freezing}                                                                                                                                                                                                                        \\
                \quad User $i \in \mathcal{U}$:                                                                                                                                                                                                                              \\
                \quad\quad - pad $\bm{x}^{i(t)}$ randomly and group the entries.                                                                                                                                           \\
                \quad\quad - add noise to $\bm{x}^{i(t)}$ via \cref{eq:noise}.                                                                                                                                           \\
                                \quad\quad - calculate key vector $\bm{k}^{i(t)}$ via \cref{eq:k_sec}.                                                                                                                                                                                     \\
                \quad\quad - calculate frozen vector $\bm{y}^{i(t)}$ via \cref{eq:y} and \cref{eq:y_sec}.                                            \\
                \quad\quad - \textcolor{red}{\underline{generate random vector $\bm{\zeta}^{i(t)}=({\zeta}^{i(t)}_1,{\zeta}^{i(t)}_2,\ldots,{\zeta}^{i(t)}_{l\lambda})$ and calculate commitment vector $\bm{c}^{i(t)}$ via \cref{eq:commitment}.}}                        \\
                \quad\quad - obtain $m_1^i=$ $\bm{y}^{i(t)}$\textcolor{blue}{\underline{(Do not send $\bm{y}^{i(t)}$ when there is RVE)}}\textcolor{red}{\underline{$||\bm{\zeta}^{i(t)}||\bm{c}^{i(t)}$}}, send $\sigma_{1}^i \rightarrow DS.sign( {{sig}_{i}^{sk},m_1^i} )$ to $\mathcal{S}$.                                                                                                                                      \\
                $~\bullet~$\textbf{Phase 2  SecAgg}                                                                                                                                                                                                                          \\
                \quad $\mathcal{S}$ and Users:                                                                                                                                                                                                                               \\
                \quad\quad - execute \textbf{SAP} for $\{\bm{k}^{i(t)}\}_{i\in \mathcal{U}}$. \\
                \quad \quad \quad \quad * $\cdots$                                                                                                                                                                                                                           \\
                \quad \quad \quad \quad * \textcolor{blue}{\underline{users get $\left(\bm{\kappa_1,\kappa_2}\right)$ and obtain $m_2^i=$ $\bm{\grave{y}}^{i(t)}( \bm{\kappa_1}\bm{y}^{i(t)}) || \bm{\acute{y}}^{i(t)}(\bm{y}^{i(t)}+\bm{\kappa_2})$, send $\sigma_{2}^i \rightarrow DS.sign( {{sig}_{i}^{sk},m_2^i} )$ to $\mathcal{S}$}}.                                                                \\
                \quad \quad \quad \quad * $\cdots$                                                                                                                                                                                                                           \\
                \quad\quad - all participants receive $\mathcal{U}'$ and $\sum_{i\in\mathcal{U}'}\bm{k}^{i(t)}$ (or $Enc(\sum_{i\in\mathcal{U}'}\bm{k}^{i(t)})$).                               \\

                \quad $\mathcal{S}$:                                                                                                                                                                                                                                         \\
                \quad \quad - if $DS.verify( {\sigma_{1}^i, {sig}_{i}^{pk},m_1^i} )\rightarrow False$, abort. Otherwise, calculate $\sum_{i\in \mathcal{U}'}\bm{y}^{i(t)}$\textcolor{blue}{\underline{(can not get $\sum_{i\in \mathcal{U}'}\bm{y}^{i(t)}$) calculate $\sum_{i\in \mathcal{U}'}\bm{\grave{y}}^{i(t)}$ and $\sum_{i\in \mathcal{U}'}\bm{\acute{y}}^{i(t)}$}}.                                                                                                         \\
                \quad \quad - \textcolor{red}{\underline{for $j\in [1,l]$, $r\in [1,\lambda]$, reveal the commitments via \cref{eq:opencommit}. If verification fails, abort.}}                                                                             \\
                $~\bullet~$\textbf{Phase 3  Thawing}                                                                                                                                                                                                                         \\
                \quad \textit{Thawing on the server side}                                                                                                                                                                                                                    \\
                \quad $\mathcal{S}$:                                                                                                                                                                                                                                         \\
                                \quad\quad - calculate $\bm{sum}=\sum_{i\in \mathcal{U}'}\bm{x}^{i(t)}_{pad}$ via \cref{eq:lineareqs} and send $\bm{sum}$ and $\sigma_{3} \rightarrow DS.sign( {{sig}_{S}^{sk},\bm{sum}} )$ to $i\in \mathcal{U}'$.                                                                                      \\
                \quad User $i \in \mathcal{U}'$:                                                                                                                                                                                                                             \\
                \quad\quad - receive $\sum_{i\in \mathcal{U}'}\bm{x}^{i(t)}$, if $DS.verify( {\sigma_{3}, {sig}_{S}^{pk},\bm{sum}} )\rightarrow False$, abort. Otherwise, unpad and output.                                                                                                                                                                                  \\
                \quad \textit{Thawing on the user side}                                                                                                                                                                                                                      \\
                \quad $\mathcal{S}$:                                                                                                                                                                                                                                         \\
                \quad\quad - send $\sum_{i\in \mathcal{U}'}\bm{y}^{i(t)}$, $Enc(\sum_{i\in\mathcal{U}'}\bm{k}^{i(t)})$ and $\sigma_{3} \rightarrow DS.sign( {{sig}_{S}^{sk},\sum_{i\in \mathcal{U}'}\bm{y}^{i(t)}}|| Enc(\sum_{i\in\mathcal{U}'}\bm{k}^{i(t)}))$ to $i \in \mathcal{U}'$.                                                                                                                   \\
                \quad User $i \in \mathcal{U}'$:                                                                                                                                                                                                                             \\
                \quad\quad - if $DS.verify( {\sigma_{3}, {sig}_{S}^{pk},\sum_{i\in \mathcal{U}'}\bm{y}^{i(t)}}|| Enc(\sum_{i\in\mathcal{U}'}\bm{k}^{i(t)} )\rightarrow False$, abort. Otherwise, \textcolor{blue}{\underline{verify $\sum_{i\in \mathcal{U}'}\bm{y}^{i(t)}$ via \cref{eq:res_verification}. If verification fails, abort.}}                                                                                                    \\
                \quad\quad - decrypt $Enc(\sum_{i\in\mathcal{U}'}\bm{k}^{i(t)})$.                                                                                                                                                                                        \\
                \quad\quad - calculate $\bm{sum}=\sum_{i\in \mathcal{U}'}\bm{x}^{i(t)}_{pad}$ via \cref{eq:lineareqs}, unpad and output.                                                                                                              \\
                {}                                                                                                                                                                                                                                                           \\
                \hline
            \end{tabular}}
    \end{center}
    \caption{The pipline of $\lambda$-SecAgg with all 4 extensions for one aggregation. \textcolor{red}{\uline{The red and underlined parts are required in the user commitment extension}}. \textcolor{blue}{\uline{The blue and underlined parts are required in the result verification extension.}}}
    \label{fig:protocol_fin}
\end{figure*}

\subsection{$\mu$-security Extension}
\label{securityex}
\textbf{Threat}. Consider an extreme scenario where an extremely powerful malicious server $\mathcal{S}$ somehow obtains an element $x^{i}_e$ of $\bm{x}^{i}$, although this is generally unrealistic in most cases. In this case, $\mathcal{S}$, during the \sysname\ (without DVE) execution, would acquire $\lambda-1$ equations of the $j$-th group where $x^{i}_e$ lies, namely $\bm{y}^{i}_j=\bm{\check{A}}\bm{d}^{i}_j$ in  \cref{eq:y}, along with $x^{i}_e$. This allows $\mathcal{S}$ to solve for all elements in the $j$-th group. Even though only a very small fraction of original vectors, $\frac{\lambda}{m}$, is leaked, we still aim to fortify against this. Assuming the adversary's capability is to steal $\mu~(< \lambda - 1)$ elements from each group of $\bm{x}^{i}$. 

We expand the definitions of incomplete matrix $\bm{\check{A}}$, and residual vector $\bm{\alpha}$ of \sysname. We define \textit{Incomplete Matrix with $\mu$-security} $\bm{\check{A}}^{\mu+1}=\bm{A}_{:\lambda-\mu-1, :}$ and \textit{Residual Matrix with $\mu$-security} $\bm{\alpha}^{\mu+1}=\bm{A}_{\lambda-\mu:, :}$. Then \textit{Frozen Vector with $\mu$-security} is:
\begin{equation}
    \label{eq:y_sec}
\bm{y}^{i}  =\left(\bm{\check{A}}^{\mu+1}\bm{d}^{i}_1,\bm{\check{A}}^{\mu+1}\bm{d}^{i}_2,\ldots,\bm{\check{A}}^{\mu+1}\bm{d}^{i}_l\right),
\end{equation}
and \textit{Key Vector with $\mu$-security} is:
\begin{equation}
    \label{eq:k_sec}
    \bm{k}^{i} =\left(\bm{\alpha}^{\mu+1}\bm{d}^{i}_1,\bm{\alpha}^{\mu+1}\bm{d}^{i}_2,\ldots,\bm{\alpha}^{\mu+1}\bm{d}^{i}_l\right).
\end{equation}
Clearly, the aforementioned changes won't affect the execution process and security of \sysname. However, the compression factor of \sysname\ will decrease from ${\lambda}$ to $\frac{\lambda}{\mu+1}$ ($0\le \mu<\lambda-1$). With these modifications, the adversary, apart from the known private elements, cannot obtain additional elements from the execution of \sysname. Particularly, in the methodology outlined in Section~\ref{mainmethod}, we consider $\mu=0$. 

\subsection{Result Verification Extension}
\label{ex:rve}
\textbf{Threat}. The malicious server might intentionally provide incorrect aggregated results to disrupt training, even if it \textbf{doesn't compromise user privacy}. When integrating \sysname\ with SAPs featuring result verification capabilities, we contemplate incorporating RVE.

Most recent SAPs that can verify aggregated results typically necessitate honest and non-dropping participants, such as \textit{Collectors}~\cite{Wang2023VOSAVA} and \textit{Auxiliary Nodes}~\cite{Eltaras2023EfficientVP}, introducing additional assumptions. Consequently, by utilizing this extension, we unavoidably add security assumptions. Without loss of generality, we adopt the approach and security assumptions from VerSA~\cite{Hahn2023VerSAVS} for \sysname, which requires the server and users not to collude but does not need other trusted third parties.
Most SAPs~\cite{aono2017privacy,Bonawitz2017PracticalSA,Bell2020SecureSA,Li2022EfficientPF,xu2022non,Liu2022EfficientDA,stevens2022efficient,Wei2023LightweightFL} lack the functionality to verify aggregated results. This extension merely constrains the server within the \sysname\ module to obtain the correct $\sum_{i\in\mathcal{U}'}\bm{y}^{i}$, and it allows secure integration of PVF and SAPs in the face of aggregation falsification attacks.

\textbf{Phase 0: MainRV.Setup($\cdot$)}. Execute Main.Setup($\cdot$).

\textbf{Phase 1: MainRV.Freeze($\cdot$)}. Execute Main.Freeze($\cdot$).

\textbf{Phase 2: MainRV.SecAgg($\cdot$)}. During the execution of VerSA, all users obtain a set of random vectors $\left(\bm{\kappa}_1,\bm{\kappa}_2\right)$ derived from shared keys. In the verification step, users send $\bm{\grave{y}}^{i}= \bm{\kappa}_1\bm{y}^{i}$ and $\bm{\acute{y}}^{i}= \bm{y}^{i}+\bm{\kappa}_2$ to the server, which simultaneously aggregates $\bm{\grave{y}}^{i}$ and $\bm{\acute{y}}^{i}$. In this process, the verification of the sum of frozen vectors and the sum of key vectors is combined, avoiding additional interaction. VerSA ensures users obtain a consistent $\mathcal{U}'$ and the correct $\sum_{i\in\mathcal{U}'} \bm{k}^{i}$. Users also receive $\sum_{i\in\mathcal{U}'} \bm{\grave{y}}^{i}$ and $\sum_{i\in\mathcal{U}'} \bm{\acute{y}}^{i}$. 

\textbf{Phase 3: MainRV.Thaw($\cdot$)}. After receiving the results, users proceed with verification:
\begin{equation}
    \label{eq:res_verification}
    \sum_{i\in\mathcal{U}'} \bm{\grave{y}}^{i} \setminus \bm{\kappa}_1 \overset{\text{?}}{=} \sum_{i\in\mathcal{U}'} \bm{\acute{y}}^{i} - |\mathcal{U}'|\bm{\kappa}_2.
\end{equation}
$\bm{y}^{i}$ is unknown to $\mathcal{S}$ throughout the process. If the verification passes, execute Main.Thaw($\cdot$) for thawing on the user side. Otherwise, conclude that the server deviates from the protocol and terminate the execution.

\subsection{User Commitment Extension}
\label{ex:uce}
\textbf{Threat}. In \sysname, if a malicious user $i$ uses inconsistent $\bm{d}^{i}_j$ and $\bm{{d}}'^{i}_j$ for freezing, resulting in obtaining $\check{\bm{A}}\bm{d}^{i}_j$ and $\bm{a}\bm{{d}}'^{i}_j$, or applies mismatched $\check{\bm{A}}$ and $\bm{\alpha}'$ to $\bm{d}^{i}_j$ for linear transformations, the frozen vector and key vector no longer correspond. Although this remains \textbf{harmless to the privacy} of honest users, it will lead to inaccuracies in the final aggregated result.

This malicious tampering is unrelated to SAPs. In the freezing phase of \sysname, we use \textit{Pedersen Commitment}~\cite{pedersen1991non} to ensure users' vector consistency. \textbf{Note} that, to enable verification, it's crucial to ensure that during the validation phase, the server is semi-honest and cannot collude with users. Otherwise, a malicious server could illegitimately approve malicious user actions, making user verification invalid, which is widely adopted in prior works~\cite{rathee2023elsa}. Below, we outline UCE to modify the Main method to achieve the aforementioned functionality. UCE supports consistency verification of the user inputs in scenarios with advanced needs. Note that UCE is only applicable to SAPs where the server has access to the aggregated results.

\textbf{Phase 0: MainUC.Setup($\cdot$)}. Execute Main.Setup($\cdot$). Given the security parameter $\rho$, it generates the group $(\mathbb{G}_p,p,g)$, where $p$ is the order of $\mathbb{G}_p$ and $g$ is its generator. $h$ is an element of $\mathbb{G}_p$. $\mathbb{G}_p,p,g,h$ are public.

\textbf{Phase 1: MainUC.Freeze($\cdot$)}. The user $i \in [1, n]$ generates the \textit{Random Vector} $\bm{\zeta}^{i} = (\zeta^{i}_1, \zeta^{i}_2, \ldots, \zeta^{i}_{l\lambda})$ and calculates the \textit{Commitment Vector}: 
\begin{equation}
    \label{eq:commitment}
    \resizebox{0.9\columnwidth}{!}{$\begin{aligned}
        \bm{c}^{i} 
                          & =\left(\left(c^{i}_{1},\ldots,c^{i}_{\lambda}\right),\ldots,\left(c^{i}_{(l-1) \lambda +1},\ldots,c^{i}_{l \lambda }\right)\right)                                                                     \\
                          & =((g^{x^{i}_1}h^{{\zeta}^{i}_1},g^{x^{i}_2}h^{{\zeta}^{i}_2},\ldots,g^{x^{i}_{\lambda }}h^{{\zeta}^{i}_{ \lambda }}),\ldots,                                                               (g^{x^{i}_{(l-1) \lambda +1}}h^{{\zeta}^{i}_{(l-1) \lambda +1}},g^{x^{i}_{(l-1) \lambda +2}}h^{{\zeta}^{i}_{(l-1) \lambda +2}},\ldots,g^{x^{i}_{l \lambda }}h^{{\zeta}^{i}_{l \lambda }}).
    \end{aligned}$}
\end{equation}
Execute Main.Freeze($\cdot$).

\textbf{Phase 2: MainUC.SecAgg($\cdot$)}. Execute Main.SecAgg($\cdot$), once completed, all users obtain the aggregated result $\bm{sum}$. Each user $i$ sends $\bm{c}^{i}$ and $\bm{\zeta}^{i}$ to the server. For $j\in [1, l]$, $r\in [1,\lambda]$, $\mathcal{S}$ validates:
\begin{equation}
    \label{eq:opencommit}
    \prod_{i\in \mathcal{U}'} c^{i}_{(j-1)\lambda+r} \overset{\text{?}}{=}g^{sum_{(j-1)\lambda+r}}h^{\sum_{i\in \mathcal{U}'}{\zeta}^{i}_{(j-1)\lambda+r}}.
\end{equation}
If the validation fails, it implies the user deviates from the protocol, and the protocol is terminated. Otherwise, it signifies that the user employs the consistent original vector when generating the frozen vector and key vector, and the utilized public parameters are accurate. 

\textbf{Phase 3: MainUC.Thaw($\cdot$)}. Execute Main.Thaw($\cdot$).

Clearly, this extension will inevitably introduce a performance decrease. It's worth noting that in the majority of SAPs, many malicious user behaviors, such as incorrectly executing key agreements or submitting counterfeited secret shares, result in aggregation errors. And strictly enforcing user behavior remains a pending issue~\cite{Ma2023FlamingoMS}. 

\subsection{Evaluation of Extensions}
\label{appendix:eva_of_extensions}

\textbf{End-to-end comparison. } In accordance with the experimental settings of \cref{fig:end2end}, we evaluate the impact of RVE, UCE, and $\mu$-security on the overall model training time (as illustrated in the \cref{fig:end2end_supp}). The acceleration effect and communication expansion when integrating different extensions are 34.5× and 1.4× (RVE), 21.1× and 2.9× (UCE), 13.5× and 1.0× (5-security).

\begin{figure}[h!]
    \centering
    \includegraphics[scale=0.32]{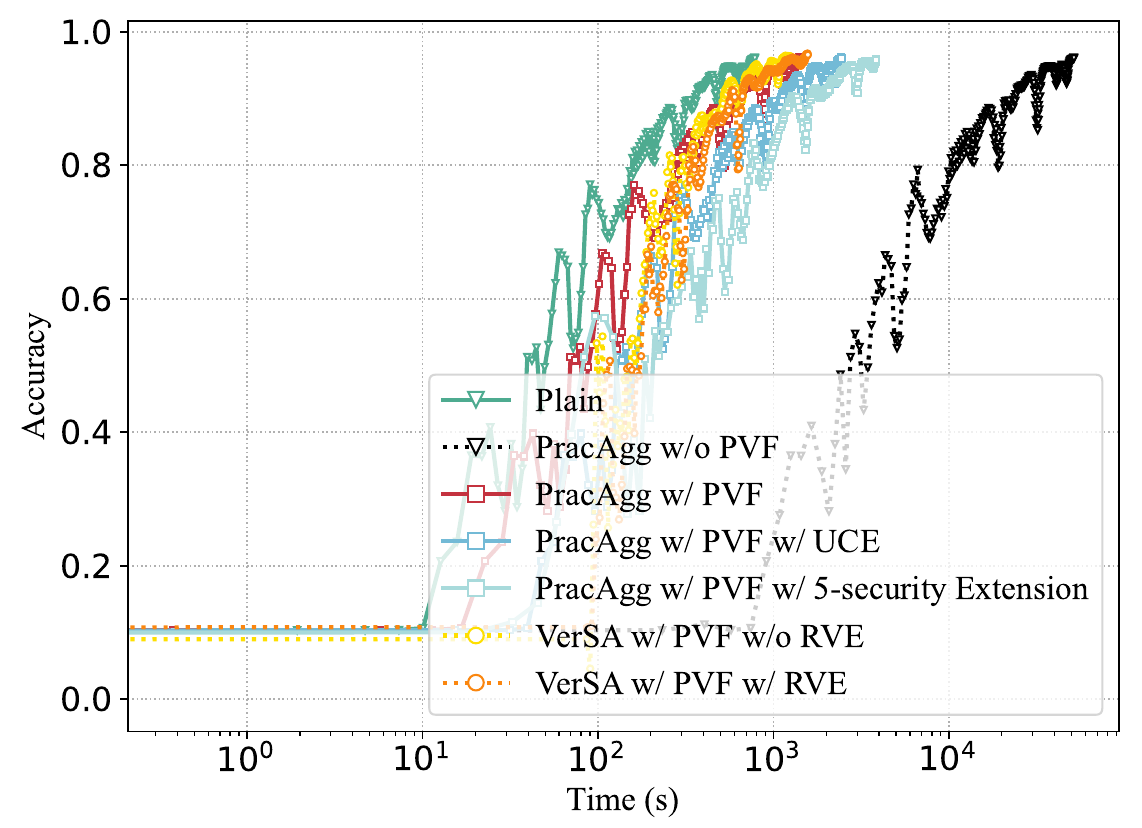}
    \caption{Comparison of various aggregation methods, with $n=100,\eta=5\%,\lambda=100$.}
    \label{fig:end2end_supp}
\end{figure}

$\bm{(\lambda,\mu)}$. We evaluate the speedup with different $(\lambda, \mu)\in\{100,300,500,700,1000\}\times[1,10]$. The integration of \sysname\ with the $\mu$-security extension reduces the entries of vectors involved in secure aggregation to $\frac{\mu + 1}{\lambda}(\mu < \lambda -1)$ of their original size. \cref{fig:lam_mu} showcases incorporating the $\mu$-security extension does diminish the improvement factor. However, even when $\mu=0.1\lambda$, \sysname\ still yields an acceleration gain of $10\times$ along with communication improvements exceeding $5\times$.

\textbf{RVE and UCE. }UCE necessitates users to commit to each dimension of the original vectors, while $\mathcal{S}$ needs to validate each dimension. RVE requires users to submit frozen vectors twice and $\mathcal{S}$ to perform summations twice. Thereby, both UCE and RVE impose on each participant an additional communication complexity of $O(m)$ and computation complexity of $O(m)$. 

Figure \ref{fig:exten} presents the overhead required for SAPs with different extensions. 
Compared to not integrating \sysname, the computation overhead after adding extensions is still nearly \textbf{an order of magnitude lower}. Communication costs from extensions primarily stem from the transmission of additional vectors, such as commitment vectors. We leave mitigating these overheads to future work.

\begin{figure*}[h!]
\centering
\includegraphics[width=0.8\textwidth]{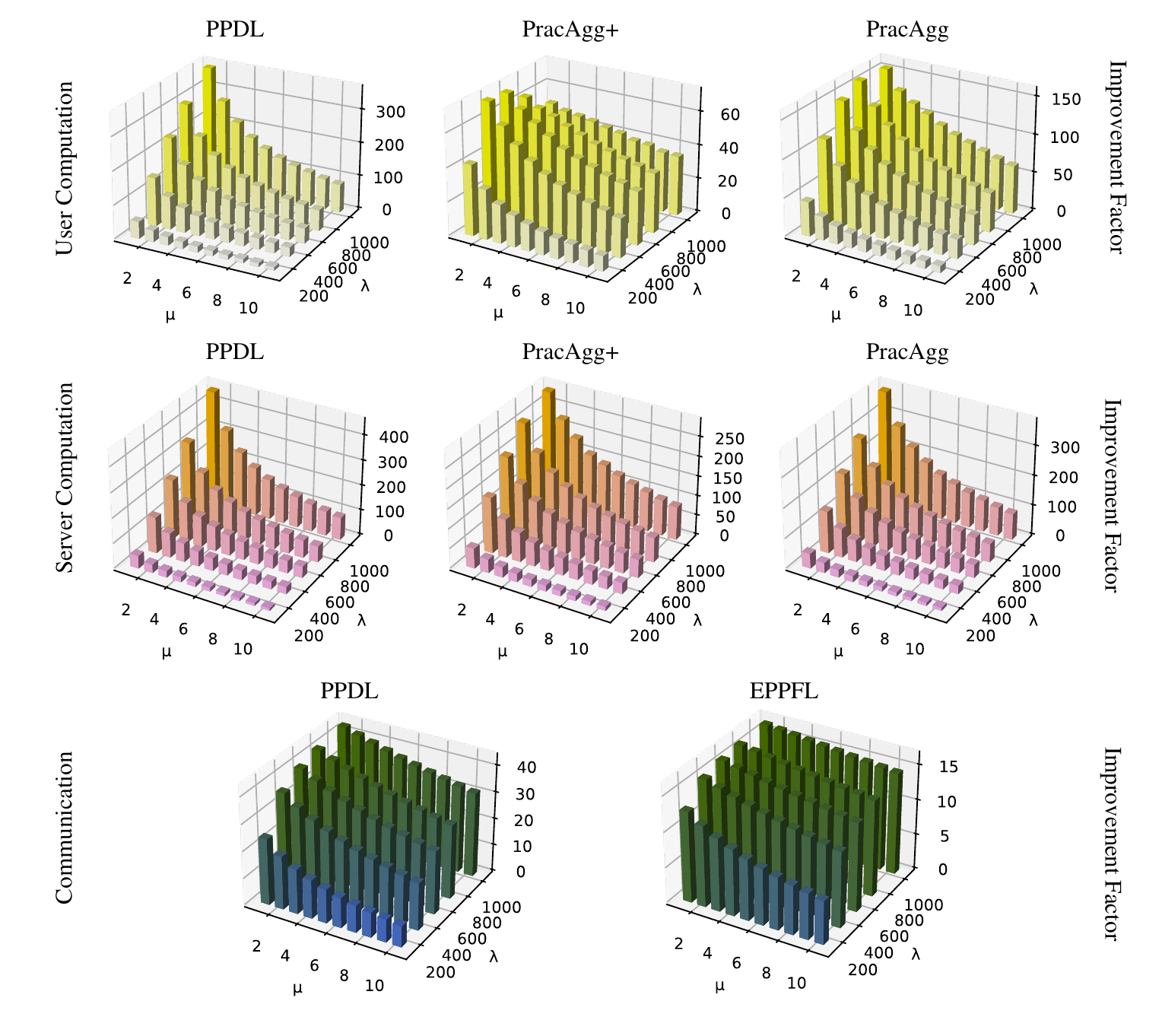}
\caption{The impact of different ${(\lambda,\mu)}$ on improvement factor. $n=100$, $m=100k$, $\lambda=100$, and $\eta=10\%$.}
\label{fig:lam_mu}
\end{figure*}

\begin{figure*}[h!]
\centering
\includegraphics[width=0.97\textwidth]{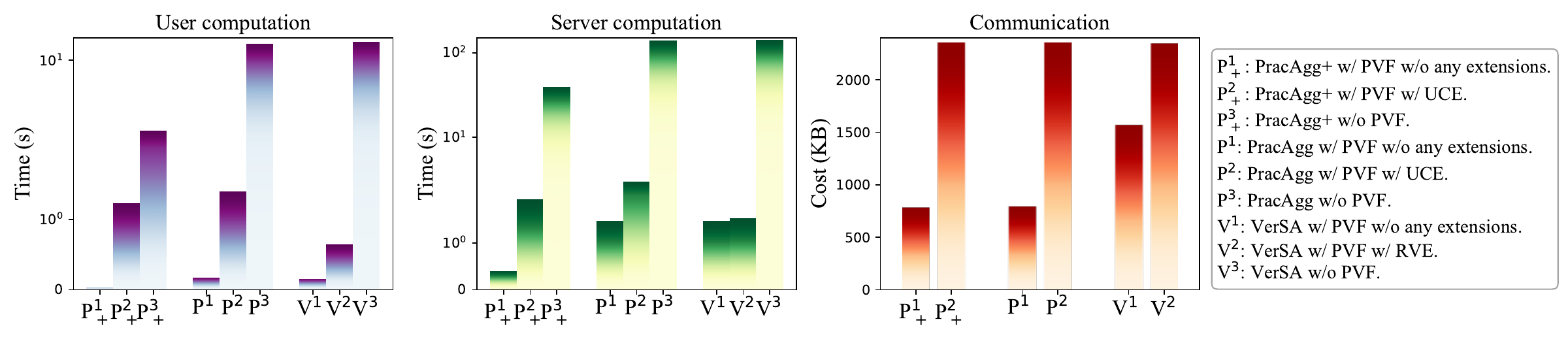}
\caption{Comparison of costs across different extensions. $n=100$, $m=100k$, $\lambda=100$, and $\eta=10\%$.}
\label{fig:exten}
\end{figure*}

\section{Portability Analysis}
\label{app:port}
\begin{table*}[!t]
    \begin{center}
        \caption{Symbolic representation of SAP execution process without and with \sysname. The symbols' meanings are as follows: $US$: User Selection; $UG$: User Grouping; $KA$: Key Agreements among Users; $SS$: Secret Sharing; $E/M$: Encryption or Masking user's original vectors; $UA$: Upload and Aggregation; $D/U$: Decryption or Unmasking vectors; $Ver$: Verification of aggregated results; $F$: Freezing process in \sysname\ Main Method, generating frozen vectors and key vectors; $T$: Thawing process in \sysname\ Main Method, deriving all entries based on aggregated entries; $RV$: Result Verification Extension of \sysname; $X'$: Specialized design for process $X$ in the corresponding protocol; $SA(\cdot)$: Operations involved in the aggregation process; $\square$: Protocol can execute in the semi-honest setting; $\blacksquare$: Protocol can execute in the active adversary setting. }
    \resizebox{\textwidth}{!}{\begin{tabular}{cccc}
                \toprule
                Type                            & Scheme                                   & w/o \sysname                                                       & w/ \sysname                                                            \\
                \midrule
                $\setminus$                    & PlainAgg                     & $[US,UA]$                                                     & $\setminus $                                                      \\
                \midrule
                \multirow{2}{*}{HE-based}       & PPDL           & $[US,\mathbf{SA}(KA,E/M,UA,D/U)]_{\square}$                   & $[US,F,\mathbf{SA}(KA,E/M,UA,D/U),T]_{\square}$                   \\
                {}                              & EPPFL         & $[US,\mathbf{SA}(KA,E/M,UA,D/U)]_{\square}$                   & $[US,F,\mathbf{SA}(KA,E/M,UA,D/U),T]_{\square}$                   \\
                \midrule
                SMPC-based                      & NIVP-DS              & $[US,\mathbf{SA}(KA,SS,E/M,UA,D/U)]_{\square}$                & $[US,F,\mathbf{SA}(KA,SS,E/M,UA,D/U),T]_{\square}$                \\
                \midrule
                \multirow{4}{*}{Mask-based}     & PracAgg & $[US,\mathbf{SA}(KA,SS,E/M,UA,D/U)]_{\square\blacksquare}$    & $[US,F,\mathbf{SA}(KA,SS,E/M,UA,D/U),T]_{\square\blacksquare}$    \\
                {}                              & PracAgg+       & $[US,\mathbf{SA}(UG,KA,SS,E/M,UA,D/U)]_{\square\blacksquare}$ & $[US,F,\mathbf{SA}(UG,KA,SS,E/M,UA,D/U),T]_{\square\blacksquare}$ \\
                {}                              & EffiAgg      & $[US,\mathbf{SA}(KA,SS,E/M',UA,D/U')]_{\square\blacksquare}$  & $[US,F,\mathbf{SA}(KA,SS,E/M',UA,D/U'),T]_{\square\blacksquare}$  \\
                {}                              & LPPFedL    & $[US,\mathbf{SA}(KA,SS,E/M',UA,D/U')]_{\square}$              & $[US,F,\mathbf{SA}(KA,SS,E/M',UA,D/U'),T]_{\square}$              \\
                \midrule
                Result Veri.                   & VerSA          & $[US,\mathbf{SA}(KA,SS,E/M,UA,D/U,Ver)]_{\square}$            & $[US,F,\mathbf{SA}(KA,SS,E/M,UA,D/U,Ver),RV,T]_{\square}$         \\
                \midrule
                \multirow{2}{*}{Multi. Privacy} & LTPA          & $[US',\mathbf{SA}(KA,SS,E/M,UA,D/U)]_{\square}$               & $[US',F,\mathbf{SA}(KA,SS,E/M,UA,D/U),T]_{\square}$               \\
                {}                              & MRSA            & $[US',\mathbf{SA}(KA,SS,E/M,UA,D/U)]_{\square}$               & $[US',F,\mathbf{SA}(KA,SS,E/M,UA,D/U),T]_{\square}$               \\
                \midrule
                Resist. M. Incon.            & $\setminus$                                & $[US,\mathbf{SA}(KA,SS,E/M',UA,D/U)]_{\square}$               & $[US,F,\mathbf{SA}(KA,SS,E/M',UA,D/U),T]_{\square}$               \\
                \bottomrule
            \end{tabular}}
        \label{tab:portability}
    \end{center}
\end{table*}
PVF does not attempt to alter SAP, and the decoupling greatly enhances the portability.
PPDL~\cite{aono2017privacy} and EPPFL~\cite{Li2022EfficientPF} employ homomorphic encryption, eliminating the need for secret sharing. NIVP-DS~\cite{xu2022non} constitutes a dual-server secure multi-party computation scheme, requiring users to share secrets between two servers. PracAgg~\cite{Bonawitz2017PracticalSA} and PracAgg+~\cite{Bell2020SecureSA} represent classic mask-based solutions, with PracAgg+ necessitating an additional user grouping process. EffiAgg~\cite{Liu2022EfficientDA} and LPPFedL~\cite{Wei2023LightweightFL} respectively introduce specialized masking mechanisms to lightweight PracAgg. VerSA~\cite{Hahn2023VerSAVS} enables users to verify the aggregated result at the conclusion, and its integration with \sysname\ requires RVE. LTPA~\cite{Liu2023LongTermPA} and MRSA~\cite{so2023securing} individually design specific user selection mechanisms to ensure privacy in \textbf{multi-round aggregation}. Resistance against \textbf{model inconsistency attacks} can be achieved by making slight modifications to PRG without incurring additional overhead~\cite{Ma2023FlamingoMS}. PVF also fits the \textbf{asynchronous} setting. Since PVF itself is one-shot and decoupled from specific SAP, it does not affect the one-shot masking or recovery in asynchronous SAP (such as LightSecAgg~\cite{so2022lightsecagg}).
We symbolically represent the entire process of federated learning aggregation in Table \ref{tab:portability}. It is evident that \sysname\ is decoupled from SAPs, not interfering with the internal execution process of SAP. 

\end{document}